\documentclass[10pt, oneside]{article}
\usepackage[all]{xy}
\usepackage{amssymb}
\usepackage{amsthm}
\usepackage{amsmath}
\usepackage{newlfont}
\usepackage[T1]{fontenc}
\usepackage{mathrsfs}
\usepackage{vmargin}
\usepackage{graphicx}
\usepackage{verbatim}
\usepackage{float}
\usepackage{multirow}
\usepackage{setspace}
 \usepackage[bf]{caption}
 \usepackage{enumerate}
 \usepackage{epstopdf}

\usepackage{vmargin}
    \setpapersize{A4}
    \setmarginsrb{32mm}{15mm}{32mm}{15mm}
                 {7mm}{7mm}{11mm}{11mm} 

\usepackage[latin1]{inputenc}
\usepackage[english]{babel}
\usepackage{graphicx}
\numberwithin{equation}{section}
\theoremstyle{plain}

\theoremstyle{plain}
\newtheorem*{ThmA}{\bf{Theorem A}}
\newtheorem*{ThmB}{\bf{Theorem B}}

\theoremstyle{plain}
\newtheorem{thm}{Theorem}[section]
\newtheorem{corollary}[thm]{Corollary}

\newtheorem{proposition}[thm]{Proposition}
\theoremstyle{definition}
\newtheorem{definition}[thm]{Definition}
\theoremstyle{remark}

\newcommand {\e}{\mathbb E}
\newcommand{\VolX}{\langle X \rangle _t}

\newcommand{\VolY}{\langle Y \rangle _t}

\newcommand{\target}{\overline \sigma}

\newcommand{\I}{1{\hskip -2.5 pt}\hbox{I}}

\def\registered{{\ooalign{\hfil\raise .00ex\hbox{\scriptsize R}\hfil\crcr\mathhexbox20D}}}

\begin{document}

\title{\bf{ Valuation of asset and volatility derivatives using decoupled time-changed  L\'{e}vy processes} \footnote{A MATHEMATICA\textsuperscript{\textregistered} online supplement to this paper containing numerical reuslts is available at the website: \texttt{http://lorenzotorricelli.it/Code/DTC\_TVO\_Implementation.nb}.}  }
\author{Lorenzo Torricelli\footnote{Department of Mathematics,
University College London, \texttt{lorenzo.torricelli.11@ucl.ac.uk }} \date{}}

\maketitle

\vspace{-.3cm}

\begin{abstract}
In this paper we propose a general derivative pricing framework that employs \emph{decoupled time-changed} (DTC) {L\'{e}vy processes} to model the underlying assets of  contingent claims. 
 A DTC L\'{e}vy process  is a generalized time-changed L\'evy process whose continuous and pure jump parts are allowed to follow separate random time scalings;
  we devise the martingale structure for a DTC L\'evy-driven asset and revisit  many popular models which fall under this framework. Postulating different time changes for the underlying L\'evy decomposition allows the introduction of asset price models consistent with the assumption of a correlated pair of continuous and jump market activity rates; we study one illustrative DTC model 
of this kind based on the so-called \emph{Wishart process}.   The theory we develop is applied to the problem of pricing not only claims that depend on the price or the volatility of an underlying asset, but also more sophisticated derivatives whose payoffs rely on the joint performance of these two financial variables, such as the \emph{target volatility option} (TVO).  We solve the pricing problem through a Fourier-inversion method. Numerical analyses validating our techniques are provided. In particular, we present some evidence that correlating the activity rates could be beneficial for modeling the  volatility skew dynamics.  
\end{abstract}


\section{Introduction}
The use of L\'{e}vy models in finance dates back to to the classic work of Merton (1976), who proposed that the log-price dynamics  of a stock return should follow an exponential Brownian diffusion punctuated by a Poisson arrival process of normally distributed jumps. In that work, two of the main shortcomings of the Black-Scholes model, the continuity of the sample paths and the normality of returns, were addressed for the first time. Over the years, L\'{e}vy processes have proved to be a flexible and yet mathematically tractable instrument for asset price modeling and sampling. 
One of the easiest ways of producing a L\'{e}vy process is to use the principle of \emph{subordination} of a Brownian motion $W_t$. If $T_t$ is an increasing L\'{e}vy process independent of $W_t$, then the subordinated process $W_{T_t}$
will still be of L\'{e}vy type. 
 Subordination is the simplest example of a \emph{time change}, that is, the operation whereby one considers the time evolution of a stochastic process as occurring at a random time.
 
 Return models depending on time-changed Brownian motions have been conjectured since Clark (1973); further theoretical support to the financial use of time-changed models is given by Monroe's (1978) theorem, asserting that any semimartingale can be viewed as a time change of a Brownian motion. Consequently, any semimartingale representing the log-price process of an asset can be considered as a re-scaled Wiener process. Empirical studies (Ane and Geman, 2000) confirmed that normality of returns can be recovered in a new price density based on the quantity and arrival times of orders, which justifies the interpretation of $T_t$ as ``\emph{business time}'' or  ``\emph{stochastic clock}''; the instantaneous variation of $T_t$ is hence the ``\emph{activity rate}'' at which the market reacts to the arrival of information. Further advances were made by Carr and Wu (2004),  who demonstrated that much more general time changes are potential candidates for asset price modeling, and effectively recovered many models from the standard literature by using a time-changed representation.
 
 However, not all the possibilities in time change modeling have been exhausted by the current research. For example, the stochastic volatility model with jumps (SVJ)  treated among the others by  Bates (1996), and the  stochastic volatility model with jumps and stochastic jump rate (SVJSJ) studied by  Fang (2000), although retaining a time re-scaled structure, are not time-changed L\'{e}vy processes as they are understood in Carr and Wu (2004). Indeed, in these two classes of models the jump component does not follow the same time scaling as the continuous Brownian part: in the SVJ model the discontinuities have stationary increments, whereas in the SVJSJ model the jump rate is allowed to follow a stochastic process of its own. In other words, price models for which the ``stochastic clock'' runs at different paces for the ``small'' and ``big'' market movements have already been proposed and tested. The statistical analyses of Bates (1996) and Fang (2000) confirm that these models are capable of an excellent data fitting, in particular the SVJSJ model.
 As pointed out by  Fang (2000), there are various other reasons for conjecturing a stochastic jump rate. If activity rates are to be interpreted as the frequencies of arrival of new market information, it seems unlikely that such rates could be taken as constant, as this would imply a constant information flow. Moreover, a constant jump rate implies stationary jump risk premia, which also seems unreasonable. Another stylized fact potentially captured by a model with a stochastic jump rate is the slow convergence of returns to the normal distribution, which is not a feature of stationary jump models.  Despite all these considerations, the idea of a stochastic jump rate has never really caught on.

On the other hand,  if we want to exogenously model the market activity, the hypothesis of independence between the jump and the continuous instantaneous rates  as assumed by Fang (2000) seems to be overly simplistic, as in reality the two corresponding information flows  
  may very well influence each other.  For example, a crash or soaring of the market certainly impacts the day-to-day volume of trading in the days following such an event. Conversely, a sustained high activity trend over a long period, typically associated to falling prices, may eventually lead to a sudden, panic-driven plunge in the shares' value. These and similar scenarios provide heuristic arguments for the assumption of a \emph{correlated} pair of activity rates; nevertheless, to the best of this author's knowledge, asset price models capturing this feature are not yet present in the literature.

Motivated by these arguments,  the natural question that arises is whether it is possible to manufacture consistent general time-changed price processes in which the continuous and discontinuous parts of the underlying L\'{e}vy model follow two \emph{different}, possibly correlated, stochastic time changes. We shall show that the answer is affirmative. 
The family of stochastic processes we investigate is that obtained by time-modifying the continuous and jump parts of a given  L\'{e}vy process $X_t$ by two, in principle dependent, stochastic time scalings $T_t$ and $U_t$ satisfying a certain regularity condition (definition \ref{continuity}). We call such processes \emph{decoupled time changes}. In a formula:
\begin{equation} X_{T,U}:=X^c_{T_t}+X^d_{U_t},  \end{equation}
where $X^c_t$ and $X^d_t$ represent respectively the Brownian and jump components of $X_t$. 
 
The decoupled time-changed (DTC) approach suggested allows to embed in a unifying mathematical framework  many previously-known models or classes of models, so that the DTC theory offers a natural generalization of some of the extant asset modeling research. In addition, the assumption of a pair of dependent activity rates can be captured by making use of decoupled time changes. To our knowledge, this last feature is new to the asset modeling literature. In section 7 we shall illustrate a practical example of a model having this property by considering an explicit asset evolution based on a multivariate version of the square-root process  known as the \emph{Wishart process} (e.g. Bru 1991; Gourieroux 2003; da Fonseca \emph{et al.} 2007), which we use to model the instantaneous activity rates. In section \ref{numericalwish} we provide some descriptive analysis showing that this model retains an increased flexibility for the purpose of modeling the volatility skew, compared to some popular existing jump asset price models.

A prior study supporting the financial use of DTC L\'evy is given by the work of Huang and Wu (2004). The authors conduct a specification analysis of an SDE whose solution is equivalent to a DTC L\'evy-based asset evolution as defined in this paper, give an overview of the ``nesting'' of models allowed by this setup, and then discuss the impact of various time change specifications in parameter estimation. However, their work does not provide any theoretical justification for the martingale property of the general asset price equation used. Furthermore, they do not explore the issue of dependence between the activity rates as related to analytical tractability, a natural ground of analysis provided by the model.  Indeed, to ensure the existence of a semi-closed pricing equation for the ``SV4'' model in section E, the authors have to revert to a model with independent activities. By providing a general theoretical framework for DTC-based L\'evy models, and devising an analytical DTC specification with true dependence between the stochastic volatility and the jump rate, this paper addresses both of these shortcomings.

\bigskip

From the perspective of the valuation of financial derivatives, the aim of this work is to gain some understanding of the impact on derivative pricing of the interactions between the volatility and the price of the underlying.
To give an example, a recent market innovation is that of derivatives and investment strategies based on volatility-modified versions of plain vanilla products.
Such contracts are able to replicate classic European payoffs under a perfect volatility foresight; at the same time, the component of the price that is due to a vega excess may be reduced by using the realized volatility as a normalizing factor. One example of such a product is the \emph{target volatility option}. A target volatility (call) option (TVO) pays at maturity $t$ the amount: 
\begin{equation} F(S_t, RV_t)=\frac{\target}{\sqrt {RV_t}} (S_t-K)^+,
  \end{equation}
for a strike price $K$ and  a \emph{target volatility} level $\target$, a constant that is written in the contract. Intuitively, the closer the realized volatility $RV_t$ is to $\target$, the more this claim will behave like a call option; however, the presence of $RV_t$ in the denominator decreases the sensitivity of $F$ to a change in volatility. 
It can be shown (Di Graziano and Torricelli,  2012) that the price of an at-the-money TVO is approximately that of an at-the-money Black-Scholes call having implied volatility $\target$; such a constant thus represents the subjective volatility view of an investor, which may very well differ from the spot volatilities implied by the market. 

 In view of this increasing interaction between volatility and stock in the financial assets available in the market, being able to efficiently price derivatives like the TVO and other similar products is gaining relevance. The pricing problem of hybrid volatility/asset derivatives, with special emphasis on the target volatility option, has  already been addressed by Di Graziano and Torricelli (2012) for a zero-correlation stochastic volatility model, and by Torricelli (2013), for a general stochastic volatility model. However,  to our knowledge, a comprehensive pricing framework comparable to those available for plain vanilla derivatives (e.g Carr and Madan 1999; Lewis 2000; Lewis 2001; Carr and Wu 2004) has not yet been developed: this is one limitation we intend to overcome with this paper. 
 The pricing technique we use is a well-known approach yielding a semi-closed analytical formula for the derivative price through an inverse Fourier integral.
It should be apparent that in all the models we shall investigate there is no particular reason not to consider mixed price and volatility payoffs as the default input of pricing models e.g. for numerical implementation, as the introduction of the realized volatility does not cause the Fourier-inversion technique to break down. Clearly, pricing both vanilla and pure volatility derivatives  is still possible within this framework, since the corresponding payoff types can be regarded as particular cases of our more general setting.

\bigskip

The remainder of the paper is organized as follows. In section 2 we lay out the assumptions; in section 3 we derive martingale properties for a decoupled time-changed L\'{e}vy model. Section 4  shows the fundamental relation linking the characteristic function of the log-price and its quadratic variation and the joint Laplace transform of the time changes as computed in an appropriate measure.  Section 5 is dedicated to the derivation of a pricing formula for products whose payoffs depend jointly on $S_t$ and $TV_t$. We devote section 6 to characterizing the DTC structure of a number of known models and computing the joint characteristic function discussed in section 4 for each such model. In section 7 we introduce an exemplifying model of DTC type featuring correlation between the time changes/activity rates. In section 8 we implement our formulae to valuate different asset and volatility derivatives under various market conditions and asset price models.  In this numerical section we also perform a sensitivity analysis of the model introduced in section 7 with respect to a correlation parameter. Finally, in section 9 we briefly summarize our work. The more technical proofs have been placed in the appendix.

\section{Assumptions and notation}

As customary, our market is represented by a filtered probability space $(\Omega, \mathcal F, \mathcal F_t, \mathbb P)$ satisfying the usual conditions. Throughout the paper we will assume that there
exists a money market account process paying a constant interest rate $r$.

Let $S_t$ be a non dividend-paying market asset. $\tilde{ S_t}$ will denote its time-zero discounted value $e^{-rt}S_t$. The \emph{total realized variance} on $[0, t]$ of $S_t$ is by definition the \emph {quadratic variation} of the natural logarithm of $S_t$, that is:
\begin{equation}\label{total}
TV_t:=\langle \log S \rangle_t=\lim _{|\pi| \rightarrow 0} \sum_{t_i \in \pi } | \log S_{t_{i+1}} -\log S_{t_{i}} |^2.
\end{equation}
The limit runs over the supremum norm of all the possible partitions $\pi$ of $[t_0, t]$. The \emph{total realized volatility} is $\sqrt{TV_t}$. The \emph{period realized variance} and \emph{volatility} (or realized variance/volatility \emph{tout court}) are given respectively by $RV_t=TV_t/t$ and $\sqrt{RV_t}$.
If $X_t=\log S_t$ is a semimartingale, by taking the limit in (\ref{total}) it is easy to check that:
\begin{equation} \langle X \rangle_t= X^2_t-2 \int_{0}^t X_{u^-}dX_u.  \end{equation}
The algebra of the square matrices of order $n$ with real entries is indicated by $\mathcal M_n(\mathbb R)$ and the sub-algebra of the symmetric matrices by $Sim_n(\mathbb R)$. Matrix product is denoted by juxtaposition; the scalar product between vectors is either indicated by multiplying on the left with the transposed vector $\cdot^T$ or by the usual dot notation. The symbol $Tr$ stands for the trace operator.

If $J$ is an absolutely continuous random variable, we denote by $f_J(x)$ its probability density function and by $\phi_J(z)$ its \emph{characteristic function} \begin{equation}\phi_J(z):=\e[e^{i z^T J}].  \end{equation}
For a Fourier-integrable function $f: \mathbb C^n \rightarrow \mathbb C$ its Fourier transform will be denoted $\hat f$. For a complex-valued function or a complex plane subset, $\cdot^*$ indicates the complex conjugate function or set.

When we say that a process is a martingale we mean a martingale with respect to its natural filtration. The notation for the conditional expectation of a stochastic process $X_t$ at time $t_0<  t$ with respect to $\mathcal F_{t_0}$ is $\e_{t_0}[\, \cdot \,]$. When the distribution of a process $X_t$ depends on other state variables $x_t$ (as in the case of a Markov process) the latter are implicitly understood to be given at time $t_0$ by $x_{t_0}$.
If $X_t$ is a process admitting conditional laws, the space of the integrable functions in the $t_0$-conditional distribution of $X_t$ at time $t_0<t$ is indicated by $L^1_{t_0}(X_t)$. The notation for the bilateral Laplace transform of the distribution of $X_t$ conditional on $t_0 <t$ is: 
\begin{equation}
\mathcal L_{X}(z)=\mathbb E_{t_0}[e^{-z^T X_{t}}]
\end{equation}
where for brevity we drop the dependence on $t$ and $t_0$ on the left hand side. The stochastic process of the left limits of $X_t$ is indicated $X_{t^-}$. The symbol $\Delta X_t$ stands for the difference $X_{t}-X_{t^-}$ or $X_t-X_{t_0}$ for some prior time $ t_0< t$. Equalities are always understood to hold modulo almost sure equivalence.

If $X_t$ is an $n$-dimensional  L\'{e}vy process, the \emph{characteristic exponent} of $X_t$ is the complex-valued function  $\psi_X: \mathbb C^n \rightarrow \mathbb C$ such that:
\begin{equation}\e[e^{i \theta^T X_t}]=e^{t \psi_X(\theta)}  \end{equation}
where $\theta$ lies in the subset of $\mathbb C^n$ and where the left-hand side is finite.

For a given choice of \emph{truncation function} $\epsilon(x)$ (that is, a bounded function which is {O}$(|x|)$ around 0) the characteristic exponent has the unique \emph{L\'{e}vy-Khintchine} representation:
\begin{equation}
\psi_X(\theta) = i  \mu^T_{\epsilon} \theta  -    \frac{\theta^T \Sigma \theta}{2}  + \int_{\mathbb R^n} (e^{i \theta^ T x}-1- i  \theta ^ T   {\epsilon(x)}) \nu(dx),
 \end{equation}
where $\mu_{\epsilon} \in \mathbb R^n$, $\Sigma$ is a non-negative definite $n \times n$ matrix with real-valued entries, and  $\nu(dx)$ is a Radon measure on $\mathbb R^n$ having a density function that is integrable at $+ \infty$ and O$(|x|^2)$ around 0.  We shall make the standard choice $\epsilon(x)=x \I_{|x| \leq 1}$ and drop the dependence of $\mu$ on $\epsilon$. The triplet $(\mu, \Sigma, \nu)$ is then called the \emph{characteristic triplet} or the \emph{L\'{e}vy characteristics} of $X_t$.

\medskip

A \emph{stochastic time change} $T_t$ is an $\mathcal F_t$-adapted c\`{a}dl\`{a}g stochastic process, increasing and almost surely finite, such that $T_t$ is an $\mathcal F_t$-adapted stopping time for each $t$.  
   The \emph{time change} of an $n$-dimensional  L\'{e}vy process $X_t$ according to $T_t$ is the $\mathcal F_{T_t}$-adapted process $Y_t:=X_{T_t}$.

\section{Definition, martingale relations and asset price dynamics}

 In this first section we introduce the notion of DTC L\'evy process and devise an exponential martingale structure naturally associated to it. This construct serves a twofold purpose. In first place it allows to formulate a DTC-based asset price evolution whose discounted value enjoys the martingale property. According to general theory, this in turn enables to postulate the existence of a risk-neutral measure that correctly prices the market securities. Secondly, it defines a class of complex-valued martingales pivotal for the computations of the next section.

\medskip

Let $\mathcal B$ be the space of the $n$-dimensional $\mathcal F_t$-supported Brownian motions with drift starting at $0$, and  $\mathcal J$ be the space of the $\mathcal F_t$-supported \emph{pure jump L\'{e}vy processes} starting at 0, that is, the class of the c\`{a}dl\`{a}g   $\mathcal F_t$-adapted processes with stationary and independent increments orthogonal\footnote{Two processes $X_t$ and $Y_t$ are said to be \emph{orthogonal} if $\langle X, Y \rangle_t=0$ for all $t>0$.}  to all the elements of  $\mathcal B$.

 Every L\'{e}vy process $X_t$ such that $X_0=0$ can be decomposed as the orthogonal sum 
\begin{equation} X_t= X^c_t + X^d_t,  \end{equation}
with $X^c_t \in \mathcal B$ and $X^d_t \in \mathcal J$. We shall refer to $X^c_t$ and $X^d_t$ respectively as the \emph{continuous} and \emph{discontinuous} parts of $X_t$. 

Time changes are fairly general mathematical objects, so we have to introduce some additional requirements in order for our discussion to proceed. One property we shall assume throughout is the so-called \emph{continuity with respect to the time change}.

\begin{definition}\label{continuity} Let $T_t$ be a time change on a filtration $\mathcal F_t$.  An $\mathcal F_t$-adapted process $X_t$ is said to be $T_t$-\emph{continuous}\footnote{Jacod (1979)  uses $T_t$-\emph{adapted}, and $T_t$-\emph{synchronized} is sometimes found; however, $T_t$-continuous is also common in the literature, and in our view less ambiguous.} if it is almost-surely constant on all the sets $[T_{t^-}, T_t]$.
\end{definition}

Obviously, a sufficient condition for $T_t$-continuity is the almost sure continuity of $T_t$. 
Hence, of particular relevance is the class of the \emph{absolutely continuous} time changes, with respect to which every stochastic process is continuous. Given a pair of \emph{instantaneous rate of activity} processes, that is, two exogenously-given c\`{a}dl\`{a}g positive stochastic processes $(v_t, u_t)$,
 valid time changes are given by the pathwise integrals: 
\begin{equation} \label{T} T_t=\int_0^t v_{s^-}ds,  \end{equation}
\begin{equation} \label{U}  U_t=\int_0^t u_{s^-}ds.  \end{equation}
The processes $v_t$ and $u_t$ describe the instantaneous impact of market trading and information arrival on the price, and formalize the concept  of ``business activity'' over time.

\medskip

A decoupled time change of a L\'{e}vy process is the sum of the (ordinary) time changes of its continuous and discontinuous part.

\begin{definition}\label{DTC} Let $X_t$ be an  $n$-dimensional L\'{e}vy process and $T_t$, $U_t$ two time changes such that $T_t$ is almost surely continuous and $X^d_t$ is $U_t$-continuous.  Then:
\begin{equation}X_{T,U}=X^c_{T_t} + X^d_{U_t}  \end{equation}
is the \emph{decoupled time change} of $X_t$ according to $T_t$ and $U_t$. 
\end{definition}

By (Jacod, 1979), corollaire 10.12, a first important property of $X_{T,U}$ is that it is an $\mathcal F_{T_t \wedge U_t}$ semimartingale.
 To avoid degenerate cases, in all that follows we always assume $T_t$ and $U_t$ to be such that $X^c_{T_t}$ and $X^d_{U_t}$ are Markov processes\footnote{In general, time changes of Markov processes are not Markovian; by using Dambis, Dubins and Schwarz's theorem (Karatzas and Shreve 2000,  theorem 4.6) one can manufacture a large class of counterexamples by starting from any continuous martingale that is not a Markov process.}.

\smallskip

We now define the class of exponential martingales canonically associated with $X_{T,U}$ when the time changes are absolutely continuous. The following proposition represents the main theoretical tool of this paper:

 \begin{proposition}\label{prop1} Let $X^1_t$ be an $n$-dimensional Brownian motion with drift and $X^2_t$ a pure jump L\'evy process in $\mathbb R^n$. Let $T^1_t$ and $T^2_t$ be two absolutely continuous time changes, set $X_t=X^1_t+X^2_t$ and $T_t= (T_t^1, T_t^2) $; define $X_{T_t}:=X^1_{T^1_t} + X^2_{T^2_t}$ and denote by  $\Theta \subseteq \mathbb C^n$ the domain of definition of $\e[\exp(i \theta^T  X_t)]$. The process:
\begin{equation}\label{M}M_t(\theta, X_t, T_t)=\exp \left(i \theta ^ T X_{T_t}-  T_t^1 \psi_{X^1}(\theta) - T_t^2 \psi_{X^2}(\theta) \right)\end{equation}
 is a local martingale, and it is a martingale if and only if $\theta \in \Theta_0$, where:
\begin{equation}\label{familiar}\Theta_0=\{ \theta \in \Theta \mbox{ \upshape{such that} }  \e[M_t(\theta, X_t, T_t)]=1,  \: \forall t \geq 0 \}. \end{equation}
\end{proposition}

\medskip

When $T^1_t=T^2_t$, the exponential $M_t$ reduces to an ordinary time change of the type discussed by Carr an Wu (2004). Even in this simple case proposition \ref{prop1} is not a consequence of applying Doob's optional sampling theorem to the martingale $Z_t(\theta)=\exp(i \theta^T X_t-t \psi_X(\theta))$, because the latter is not necessarily  uniformly integrable. Indeed, time-transforming a process always preserves the semimartingale property, but the martingale property is only guaranteed to be maintained for uniformly integrable martingales; 
 an actual example of an asset model of the form $Z_{T_t}$  that is a strict supermartingale was given by Sin (1998). Hence, the set $\Theta_0$ may very well trivialize to the empty set. This demonstrates that some choices of time changes are inherently unsuitable for time-changed asset price modeling. In the case of $X_{T_t}$ being a one-dimensional Brownian integral, sufficient requirements for (\ref{familiar})  to be satisfied are the well-known \emph{Novikov} and \emph{Kazamaki} conditions (Karatzas and Shreve 2000, chapter 3), under which the set $\Theta_0$ contains the whole of $\mathbb R^n$. The set $\Theta_0$ is sometimes called the \emph{natural parameter set}.

\bigskip

Having obtained martingale relations for a stochastic exponential involving $X_{T,U}$, the risk-neutral dynamics for a DTC L\'{e}vy-driven asset are defined in the usual fashion. We have the following immediate corollary to proposition \ref{prop1}:

\begin{corollary}\label{AssetDynamicsCor} 
Let $X_t$ be a scalar L\'{e}vy process of characteristic triplet $(\mu, \sigma^2, \nu)$ and $(T_t, U_t)$ a pair of absolutely continuous time changes. For a spot price value $S_0$ let, for $t>0$:\begin{equation}\label{AssetDynamics}
S_t=S_{0}\exp(r t+i \theta_0 X_{T,U}-T_t \psi^c_X(\theta_0)-U_t \psi^d_X(\theta_0))=S_{0}e^{r t}M_t( \theta_0, X_t^c + X_t^d, (T_t,U_t))
\end{equation}
with $\theta_0 \in \Theta_0$ being such that (\ref{AssetDynamics}) is a real number. The discounted process $\tilde S_t$ is a martingale, and therefore $S_t$ is a price process consistent with the no-arbitrage condition.
\end{corollary} 
 The stochastic process in (\ref{AssetDynamics}) is the fundamental asset model we shall use throughout the rest of the paper.

 \section{Characteristic functions and the leverage-neutral measure}\label{charsection}

Characteristic functions of state variables are the essential component of the Fourier-inverse pricing methodology, because state price densities are analytically available only for a small number of models; in contrast, characteristic functions are computable in closed form  in many instances (e.g. exponential L\'evy models, Ito diffusions). This effectively means that in order to compute expectations (prices), the standard approach is not to integrate a payoff against a density function, but rather the payoff's Fourier transform against the characteristic functions of the price transition densities. Famous examples include the FFT paper by Carr and Madan (1999), Lewis's book (2000) and subsequent paper (2001).

 The transform we are interested in is one associated with the price process (\ref{AssetDynamics}). Compared to the usual inverse Fourier/Laplace framework  
 the characteristic function we shall consider is not that of the discounted log-price alone, but one that incorporates also the quadratic variation of the log-process. Indeed, just as the characteristic function of the log-price  allows for the derivation of pricing formulae for contingent claims $F(S_t)$, the joint characteristic function of $\log( \tilde S_t)$ and $TV_t$ permits the valuation of payoffs of the form $F(S_t, TV_t)$.  This has been envisaged before by Carr and Sun (2007).

In the present section we compute this transform. 
There are normally two ways of computing characteristic functions/Laplace transforms of log-price densities. One is the analytical approach, which is popular for example in affine models, when the problem is ultimately reduced to solving a certain system of ODEs. The other is the probabilistic approach, in which the characteristic function of the log-price is linked with the Laplace transform of the integrated driving factors (where available) and then a change of measure is performed to keep track of correlations. As Carr and Wu (2004) show this technique is intimately connected with time-changed asset modeling; in what follows we extend it to the case of the underlying being modeled through a full DTC L\'{e}vy process.

\medskip

First of all we must verify that the quadratic variation operator respects the additivity  and time-changed structure of $X_{T,U}$. We have the following ``linearity/commutativity property'', of independent interest:

\begin{proposition}\label{linearcommutativity}
A DTC L\'{e}vy process $X_{T,U}$ is such that $X^c_{T_t}$ and $X^d_{U_t}$ are orthogonal.  Furthermore, its quadratic variation satisfies:
\begin{equation}\label{ortocomm}
 \langle X_{T,U} \rangle_t = \langle X ^c \rangle_{T_t} + \langle X^d \rangle _{U_t}=\Sigma{T_t} + \langle X^d \rangle  _{U_t}.
\end{equation}
That is, 
the quadratic variation of $X_{T,U}$ is the sum of the time changes of the quadratic variations of its continuous and discontinuous part.
\end{proposition}

Crucially, the processes $X^c_{T_t}$ and $X^d_{U_t}$ are orthogonal but not independent.   Without the $T_t$ and $U_t$-continuity assumption, this proposition would be false: a counterexample is provided in the appendix. Proposition \ref{linearcommutativity} ensures that, in presence of time continuity of the L\'evy continuous and jump parts with respect to the corresponding time changes, the quadratic variation of a DTC L\'evy process is itself of DTC-type.

 Now, for $S_t$ as in (\ref{AssetDynamics}) define:  \begin{equation}\label{rln}
  \Phi_{t_0}(z,w)=\e_{t_0}[ \exp(iz \log( \tilde S_t/S_{t_0} ) + i w  \, (TV_t-TV_{t_0} ))].
 \end{equation}
 For each $z,w$ for which the right hand side is finite,  $\Phi_{t_0}(z,w)$ is the Fourier transform is the joint transition function from time $t_0$ to time $t$ of $\log( \tilde S_t)$ and $TV_t$.  The characteristic function $\Phi_{t_0}(z,w)$ can be completely characterized in terms of the L\'evy  triplet of $X_t=X^c_t+X^d_t$ and the joint $\mathbb Q(z,w)$-distribution of $T_t$ and $U_t$ by virtue of the following proposition.

\begin{proposition}\label{DTCchar} Let $S_t$ be an asset evolution as in corollary \ref{AssetDynamicsCor}, 
and define the family of $\mathbb P$ absolutely-continuous measures $\mathbb Q(z,w) << \mathbb P $  having Radon-Nikodym derivative:
 \begin{equation}\label{measurechange}\frac{d \mathbb Q(z,w)}{d \mathbb P}=M_t((i z \theta_0,i w \theta_0), C_t + D_t, (T_t, U_t))\end{equation}
where $C_t=(X^c_t, 0)$, $D_t=(X^d_t, i \theta_0\VolX^d)$  and $M_t$ is given by (\ref{M}). For all $(z,w)$ such that $(i z \theta_0,i w \theta_0) \in \Theta_0$, the characteristic function in (\ref{rln}) is given by: 

\begin{equation}\label{generallaplace}\Phi_{t_0}(z,w)=\mathcal L_{\Delta T, \Delta U}^{\mathbb Q}  (\zeta(z,w, \mu, \sigma, \theta_0), \xi(z,w, \nu, \theta_0) ),
\end{equation}
with the notation $\mathcal L_{\Delta T, \Delta U}^{\mathbb Q}(\cdot )$ indicating the bilateral Laplace transform  of the conditional joint distribution of $T_t - T_{t_0}$ and $U_t-U_{t_0}$ taken under the measure $\mathbb Q(z,w)$, and 
\begin{align}\label{parameters} \zeta(z,w, \mu, \sigma, \theta_0) & = \theta_0 \mu  (z - i z) - \theta^2_0 \sigma^2 (  z^2 + i z     - 2 i w )/2, \\
\xi(z,w, \nu, \theta_0) & = i z \psi_X^d(\theta_0) - \psi_D(i z \theta_0,i w \theta_0).
\end{align}

\end{proposition}

Notice that unlike the density processes used for standard num\'eraire changes, the new distributions implied by (\ref{measurechange}) also accounts for the quadratic variation as a factor.  If we assume $T_t$ and $U_t$ to be pathwise integrals of the form (\ref{T}) and (\ref{U}), it is possible to interpret the Laplace transform (\ref{generallaplace}) as being the analogue of a bivariate bond pricing formula, where the short rates are replaced by the instantaneous activity rates, and the pricing measure is not given once and for all, but varies as an effect of the correlation of $(v_t ,u_t)$ with the underlying L\'evy process. The financial insight of (\ref{generallaplace}) is that it is possible to formulate a valuation theory by just modeling the joint term structure of the activity rates $v_t$ and $u_t$ and their correlation with the stock. 

Also of interest is the interpretation of the measure $\mathbb Q(z,w)$. Let us consider the special case of $X_t$ being independent of $T_t$ and $U_t$. In such a case it is straightforward to prove, by using the laws of the conditional expectation, that one obtains (\ref{generallaplace}) with $\mathbb Q(z,w)=\mathbb P$. Therefore, whenever there is no dependence between the time changes and the underlying L\'{e}vy process, no change of measure is needed in order to extract the characteristic function $\Phi_{t_0}(z,w)$. In contrast, in the presence of correlation between $X_t$ and the time changes, the family $\mathbb Q(z,w)$ gives a measurement of the impact of leverage on the price densities. Furthermore, in some well-behaved cases this change of measure can be absorbed in the $\mathbb P$-dynamics of the asset through a suitable parameter alteration of the distributions of $T_t$ and $U_t$. In accordance with Carr and Wu (2004), we call  $\mathbb Q(z,w)$ the \emph{leverage-neutral measure} and $\Phi_{t_0}(z,w)$ the \emph{leverage-neutral characteristic function}. Just as prices in a risky market can be equivalently computed in a risk-neutral environment according to a different price distribution, valuations in the presence of leverage can be performed in a different economy with no leverage by means of an appropriate distributional modification.

\section{Pricing and price sensitivities}

The characteristic function found in section \ref{charsection} is needed to obtain analytical formulae for the valuation of European-type derivatives with a sufficiently regular payoff $F$. In the present section we find a semi-analytical formula based on an inversion integral that extends the standard Fourier-inversion machinery to our multivariate context. 

Recall that since all the involved processes are Markovian, it makes sense to treat $\Phi_{t_0}(z,w)$ like a Gauss-Green integral kernel depending only on some given initial states at time $t_0$. The following proposition extends both theorem 1 of  Lewis (2000) and proposition 3.1 of Torricelli (2013):

\begin{proposition}\label{LewisThm}Let $Y_t=\log S_t$, with $S_t$ given in corollary \ref{AssetDynamicsCor}. Let $F(x,y) \in L^1_{t_0}(Y_t, \VolY)$  for all $t_0<t$,  be a positive 
 payoff function having analytical Fourier transform $\hat F(z,w)$ in a multi-strip \begin{equation}\label{LewisEq}\Sigma_F=\{ (z,w) \in  \Theta, \: \alpha_1 <\mbox{ Im}(z)< \alpha_2, \:   \beta_1<\mbox{ Im}(w)< \beta_2, \, \alpha_1, \alpha_2, \beta_1, \beta_2 \in  \overline{\mathbb R} \}.\end{equation}  Suppose further that $\Phi_{t_0}(z,w)$ is analytical in \begin{equation}\Sigma_\Phi=\{ (z,w) \in \Theta, \: \gamma_1 < \mbox{ Im}(z)< \gamma_2, \,  \eta_1< \mbox{ Im}(w)<\eta_2,  \: \gamma_1, \gamma_2, \eta_1, \eta_2 \in \overline{\mathbb R} \}\end{equation} and that $\Phi_{t_0}(z,w) \in  L^1(dz \times dw)$. If $\Sigma_F \cap \Sigma_\Phi^* \neq \emptyset$, then for every multi-line: 
\begin{equation}L_{k_1,k_2}=\{ (x+ ik_1, y+ ik_2), (x,y) \in \mathbb R^2   \} \subset \Sigma_F \cap \Sigma^*_\Phi   \end{equation}
 we have that the time-$t_0$ value of the contingent claim $F$ maturing at time $t$ is given by:
\begin{align}\e_{t_0} [e^{-r(t-t_0)} & F(Y_t, \VolY)] =  \frac{ e^{-r(t-t_0)} }{4 \pi^2} \cdot  \nonumber
 \\ \label{Lewis} & \int_{i k_1-\infty}^{i k_1 +\infty}  \int_{i k_2-\infty}^{i k_2+\infty}  e^{-i w \langle Y \rangle _{t_0}} S_{t_0}^{-i z} e^{-r(t-t_0)i z} \Phi_{t_0}(-z, -w) \hat F(z,w)dz dw.  \end{align}
\end{proposition}

\medskip

It is clear that modifying the asset dynamics specifications only acts on $\Phi_{t_0}$, whereas changing the claim to be priced only influences $\hat F$. Also, by setting either variable to 0, we are able to extract from (\ref{Lewis}) the prices of both plain vanilla and pure volatility derivatives. For example, the pricing integrals by Lewis (2000, 2001) are special cases of the above equation when $F$ does not depend on the realized volatility and $\Phi_{t_0}$ is either obtained from a diffusion or a L\'{e}vy process. Moreover, equation (3.10) of Torricelli (2013) is recovered when $S_t$ is assumed to follow a stochastic volatility model.

In addition, this representation is useful if we are interested in the sensitivities of the claim value with respect to the underlying state variables. Let us consider for instance the Delta (sensitivity with respect to the change in the value of the underlying) and Gamma (sensitivity with respect to the rate of change in the value of the underlying) of valuations performed through formula (\ref{Lewis}).  Call $I(r,t_0, z,w)$ the integrand on the right hand side of (\ref{Lewis}); by differentiating (if possible) under the integral sign and noting that $\Phi_{t_0}$ has no dependence on $S_{t_0}$ we see that:
\begin{align}\Delta_t:=\frac{\partial}{\partial S} \e_{t_0} &[e^{-r(t-t_0)}F(Y_t, \VolY)]=-\frac{ e^{-r(t-t_0)} }{4 \pi^2} \int_{i k_1-\infty}^{i k_1 +\infty}  \int_{i k_2-\infty}^{i k_2+\infty} \frac{i z}{S_{t_0}} I(r,t_0, z,w)dz dw,  
\end{align}
and
\begin{align}\Gamma_t:=\frac{\partial^2}{\partial S^2} \e_{t_0} &[e^{-r(t-t_0)}F(Y_t, \VolY)]=\frac{ e^{-r(t-t_0)} }{4 \pi^2} \int_{i k_1-\infty}^{i k_1 +\infty}  \int_{i k_2-\infty}^{i k_2+\infty} \frac{i z-z^2}{S^2_{t_0}} I(r,t_0, z,w)dz dw.  
\end{align}

\emph{Mutatis mutandis} we can repeat this argument if we want to determine the price sensitivity with respect to the quadratic variation $\VolY$. Finally, as $\Phi_{t_0}(z,w)$ could also depend on other variables (e.g. an instantaneous rate of activity $v_{t_0}$) known at time $t_0$, by calling $\nu$ one such variable we have: 
\begin{align}\mathcal V_t:= \frac{\partial}{\partial \nu }\e_{t_0} &[e^{-r(t-t_0)}F(Y_t, \VolY)] =\frac{ e^{-r(t-t_0)} }{4 \pi^2}  \cdot  \nonumber  \\ & \int_{i k_1-\infty}^{i k_1 +\infty}  \int_{i k_2-\infty}^{i k_2+\infty} e^{-i w \langle Y \rangle _{t_0}} S_{t_0}^{-i z} e^{r(t-t_0)i z} \frac{\partial \Phi_{t_0}}{\partial \nu}(-z, -w)   \hat F(z,w)dz dw. 
\end{align}
This is especially well-suited to the case in which $\Phi_{t_0}(z,w)$ is exponentially-affine in $\nu$, i.e. \begin{equation}\Phi_{t_0}(z,w)=\exp(A(z,w, t-t_0)+B(z,w, t-t_0)\nu_{t_0}),  \end{equation} for some functions $A$ and $B$, when we have:
\begin{equation} \frac{\partial \Phi_{t_0}}{\partial \nu}(-z, -w)  =B(-z,-w, t-t_0)\Phi_{t_0}(-z, -w).  \end{equation}
In section 6 we shall explicitly calculate $\Phi_{t_0}$ for a number of decoupled time-changed models.

\section{Specific model analysis}
 
We now determine the DTC L\'{e}vy structure (\ref{AssetDynamics}) of various popular asset price processes, and find for each of them the corresponding leverage-neutral characteristic function $\Phi_{t_0}(z,w)$. Such a derivation allows for the full implementation of equation (\ref{Lewis}) for the pricing of joint asset and volatility derivatives in all the cases we deal with. What the discussion below should make apparent is that decoupled time changes offer a natural unifying framework for \emph{a priori} different strains of financial asset models (e.g. continuous/jump diffusions, jump diffusions with stochastic volatility, L\'{e}vy processes). By classifying models through their DTC structure it is possible to recognize a ``nesting'' pattern linking different models, in which some can be considered particular cases of some others. This is of use for numerical purposes: as we shall see in section \ref{numerical}, one single implementation of equation (\ref{Lewis}) can produce values for several models,  each one obtained by using a different instantiation of the code.  Four categories of asset models are discussed: standard L\'{e}vy processes,  stochastic volatility models, DTC jump diffusions and general exponentially-affine asset models. Throughout this section we assume $\theta_0=-i$ in (\ref{AssetDynamics}), so that $i \theta_0=1$ and all of the involved processes are real-valued. The domain $\Theta_0$ where the price processes are martingales is the whole complex plane, provided that the stochastic time changes and the underlying L\'evy components are sufficiently well-behaved, in the sense of the usual theory (e.g. Novikov condition for the stochastic variance, decay of the jump distributions, integrability conditions on the L\'evy measure etc.)



\subsection{L\'evy processes}\label{Levy}

In case of the L\'evy process the DTC structure coincides with the underlying L\'evy process. To determine $\Phi_{t_0}(z,w)$ no change of measure is necessary, so this function represents the joint conditional characteristic function of the log-price and its quadratic variation as given in the risk-neutral measure. Below, are provided the calculations for some popular models. 
 
\subsubsection{Black-Scholes model}

The classic SDE with constant parameters $\sigma, r$ driven by a Brownian motion $W_t$:
\begin{equation}dS_t=r S_t dt+\sigma S_t d W_t  \end{equation}
can be trivially recovered from (\ref{AssetDynamics}) by setting the triplet for the underlying  L\'{e}vy process $X_t=X^c_t$ to be $(0, \sigma,0)$ and letting $T_t=t$, $U_t=0$, so that $X_{T, U}=X_t$. From (\ref{generallaplace}), we immediately have:
\begin{align}\label{BS}
\Phi_{t_0}(z,w)= & \exp(-(t-t_0)\sigma^2(z^2 +i  z - 2 iw )/2).
\end{align}

\subsubsection{Jump diffusion models}\label{merton}

In their classic works, Merton and Kou (1976, 2002) proposed modeling the log-price dynamics as a finite-activity jump diffusion. 
The risk-neutral asset dynamics are given by:
\begin{equation}\label{explevy}dS_t=r S_{t^-} dt  + \sigma S_{t^-} d W_t + S_{t^-} (\exp(J)-1) d N_t  - \kappa \lambda S_{t^-} dt   \end{equation}
where $W_t$ is a standard Brownian motion, $N_t$ is a Poisson counter of intensity $\lambda$, and $J$ is the jump size distribution. $N_t$ and $W_t$ are assumed to be independent, and the compensator $\kappa$ equals $\phi_J(-i)-1.$ For the discounted price $\tilde S_t$ to be a true martingale, conditions on the asymptotic behavior of $f_J(x)$ must be imposed (see e.g. Cont and Tankov, 2003).
In the Merton model $J$ is normally distributed $J \sim \mathcal N( m, \delta^2)$, whereas Kou assumed for it an asymmetrically skewed double-exponential distribution, that is, the density function $f_J(x)$ as given by:
\begin{equation}
f_J(x)=\begin{cases}  \alpha pe^{-\alpha x}  \mbox{ if } x \geq 0 \\
  \beta q e^{ \beta x}  \mbox{ if } x < 0
\end{cases} 
\end{equation}
for $\alpha>1, \beta>0$ and $p+q=1$.

In these models no time change is involved, so $X_{T,U}$ coincides with the underlying L\'{e}vy process $X_t$ having characteristic triplet $(0, \sigma^2, \lambda f_J(x) dx)$.  To completely characterize $\Phi_{t_0}(z,w)$, observe that 
 $(X^d_t, \langle X^d \rangle_t)$ is just a bivariate compound Poisson process of joint jump density $f_{J,J^2}(x,y)$ and intensity $\lambda$, whence:
\begin{equation}\psi_D(z,w)= \lambda(\phi_{J,J^2}(z,w)-1),  \end{equation}
where $\phi_{J,J^2}(z,w)$ is the joint characteristic function of $J$ and $J^2$. We conclude from (\ref{generallaplace}) that $\Phi_{t_0}$ has the exponential structure:
\begin{align}\label{mertonPhi}\Phi_{t_0}(z,w)
 = & \exp(-(t-t_0)(\sigma^2(z^2/2 +iz/2 -2iw)/2 +\lambda( iz \kappa -\phi_{J,J^2}(z,w) + 1) ).
\end{align}
Now for the Merton model we have 
\begin{align}\label {mertonJumps}& \phi_{J,J^2}(z,w)
= \frac{\exp\left(\frac{i m z - \delta^2 z^2/2+ i m^2 w}{1-2 i \delta^2 w} \right)}{\sqrt{1-2 i \delta^2 w}},
\end{align}
and the integral converges for Im$(w)>-1/2 \delta^2$. For the Kou model we can write:
\begin{equation}\phi_{J,J^2}(z,w)=\phi_{J_+,J_+^2}(z,w)+\phi_{J_-,J_-^2}(z,w); \end{equation}
the characteristic function of the positive and negative parts are:
\begin{align}\phi_{J_+,J_+^2}(z,w)= 
  & \alpha  p \sqrt{\pi} e^{-\frac{( \alpha- iz)^2}{4 i w}} \left( \frac{ \mbox{ Erfc} \left( \frac{\alpha-iz}{2 \sqrt{-i w}} \right)}{ 2 \sqrt{-i w}}  \right), \\
 \phi_{J_-,J_-^2}(z,w)  = & \beta  q \sqrt{\pi} e^{-\frac{( \alpha- iz)^2}{4 i w}} \left( \frac{ \mbox{ Erfc} \left( \frac{\beta-iz}{2 \sqrt{-i w}} \right)}{ 2 \sqrt{-i w}}  \right), 
\end{align}
which both converge for Im$(w)>0$.

\subsubsection{Tempered stable L\'{e}vy and CGMY}

Another way of obtaining  L\'{e}vy distributions for the asset price is to directly specify an infinite activity  L\'{e}vy measure $\nu(dx)$. 
In such a case we have $X_{T,U}=X_t=X_t^d$,  with $X_t$ being a pure jump L\'{e}vy process of L\'{e}vy measure $\nu(dx)$. The two instances we analyze here are the tempered stable L\'{e}vy process (e.g. Cont and Tankov 2003), and the CGMY (Carr \emph{et al.} 2002) models. Both of these are obtained as an exponential smoothing of stable distributions; the latter can be viewed as a generalization of the former allowing for an asymmetrical skew between the distribution of positive and negative jumps. The  L\'{e}vy density for a CGMY process is:
\begin{equation}
\frac{d \nu(x)}{dx}=\frac{c_- e^{-\beta_- |x|}}{|x|^{1+ \alpha_-}}\I_{\{x<0\}}+ \frac{c_+ e^{-\beta_+ x}}{x^{1+ \alpha_+}}\I_{\{x \geq 0\}}.
 \end{equation}
which is well defined for all $c_+, c_-, \beta_+, \beta_- >0$, $\alpha_+, \alpha_-<2$. When $\alpha_+=\alpha_-$ one has the tempered stable process. For simplicity in what follows we assume  $\alpha_+, \alpha_- \neq 0,1$; for such values the involved characteristic functions still exist, but lead to particular cases. Since
 \begin{equation}\Phi_{t_0}(z,w)=\exp((t-t_0)\xi(z,w,\nu(x)dx, -i )) 
\end{equation}
to fully characterize $\Phi_{t_0}(z,w)$ we only need to determine $\psi^d_X(\theta)$ and $\psi_D(z,w)$. Letting $\gamma_1=\int_{-1}^{1} x d \nu(x)$,  the exponent $\psi^d_X(\theta)$ is given by the standard theory (Cont and Tankov 2003, proposition 4.2) as:
\begin{align}\psi^d_X(\theta)=&\gamma_1+\Gamma(-\alpha_+)\beta_+^{\alpha_+}c_+ \left( \left(1-\frac{i \theta}{\beta_+} \right)^{ \alpha_+}- 1 + \frac{i \theta \alpha_+}{\beta_+} \right)+ \nonumber \\  & \Gamma(-\alpha_-)\beta_-^{\alpha_-}c_- \left( \left(1+ \frac{i \theta} {\beta_-} \right)^{\alpha_-}- 1 - \frac{i \theta \alpha_-}{\beta_-} \right).   \end{align}

Set  $\gamma_2= \int_{-1}^1 x^2 d \nu(x)$;
the positive part $\psi^{+}_D$ of $\psi_D$ is then seen to be:
\begin{align} & \psi^{+}_D(z,w)= i z \gamma_1 + i w \gamma_2  +\int_{0}^{+ \infty} (e^{i z x + i w x^2}-1- (i z x + i w x^2))\frac{ c_+ e^{-\beta_+ x}}{x^{1+\alpha_+}} dx=   iz \gamma_1+i w \gamma_2+  \nonumber \\  & 
 i c_+ \beta_+^{\alpha_+}  \left(-w \frac{\Gamma(2-\alpha_+)}{2 i \beta_+^2} -z \frac{\Gamma(1-\alpha_+)}{2 i \beta_+} + i \Gamma(-\alpha_+) \right)  - c_+(\beta _+- i z)^{\alpha_+}  \left(\frac{i (\beta_+-iz)^2}{w} \right)^{-\alpha_+/2}   \cdot  \nonumber \\   & \left( \sqrt{\frac{i(\beta_+-iz)}{w}}  \Gamma \left(\frac{1}{2}-\frac{\alpha_+}{2} \right) \phantom{}_1 F_1\left[\frac{1-\alpha_+}{2}, \frac{3}{2} ,  \frac{i(\beta_+-iz)^2}{4w} \right]-   \right. \nonumber \\   \label{line2} & \phantom{000000000000000000000000000000000000000}  \left. \Gamma  \left(-\frac{\alpha_+}{2} \right) \phantom{}_1F_1\left[-\frac{\alpha_+}{2}, \frac{1}{2}, \frac{i(\beta_+-iz)^2}{4w} \right]  \right).
\end{align}
Here $\Gamma$ is the Euler Gamma function and $\phantom{}_1F_1$ the confluent hypergeometric function. The multi-strip of convergence of (\ref{line2}) is the set $\Sigma_\Phi= \{(z,w), \mbox{ Im}(w)> 0, \mbox{ Im}(z)>-\beta_+ \}$. The determination $\psi^{-}_D$ has a similar expression.

\subsection{Stochastic volatility and the Heston model}\label{stochvol}

In a stochastic volatility model the asset process is given, in a risk neutral-measure, by the SDE 
\begin{equation}\label{Heston}
 dS_t=r S_t dt + \sqrt{v_t}S_t dW^1_t 
\end{equation}
where $v_t$ is some continuous stochastic variance process. By the Dubins and Schwarz's theorem 
 any continuous martingale $M_t$ can be written as  $M_t= W_{\langle M \rangle_t}$ for a certain Brownian motion $W_t$, which implies that the DTC structure of a stochastic volatility model corresponds to a standard Brownian motion $W_t$ time-changed by $T_t$ as in (\ref{T}).
 In order to explicitly express the characteristic function $\Phi_{t_0}(z,w)$  we must make a specific choice for the dynamics in (\ref{Heston}). For instance, we can  make the popular choice of selecting a square-root (CIR) equation for the instantaneous variance:
\begin{equation}\label{HestonDyn}
 d v_t=\alpha(\theta-v_t)dt + \eta \sqrt{v_t}d W^2_t
\end{equation}
for positive constants $\alpha, \theta, \eta$ and a Brownian motion $W^2_t$ linearly correlated with $W^1_t$ through a correlation coefficient $\rho$. For $S_t$ to be well-defined, the parameters $ \alpha, \theta$, and $\eta$ need to satisfy the \emph{Feller condition} $2 \alpha \theta \geq \eta^2$. The system of SDEs (\ref{Heston})-(\ref{HestonDyn}) is the model by Heston (1993). As we change to the measure $\mathbb{Q}(z,w)$, the application of the complex-plane version of Girsanov's theorem and a simple algebraic manipulation reveals that the leverage-neutral dynamics $v_t^z$ of $v_t$ are of the same form as (\ref{HestonDyn}), but with parameters:
\begin{equation}\label{hestonleverage1}
 \alpha^z=\alpha-i \rho z \eta,
\end{equation}
\begin{equation}\label{hestonleverage3}
\theta^z= \alpha \theta/ \alpha^z
\end{equation}
(see also Carr and Wu 2004). Using equation (\ref{generallaplace}), we determine $\Phi_{t_0}$ as follows\footnote{Torricelli (2013) has independently found $\Phi_{t_0}$ for the Heston model by augmenting the SDE system (\ref{Heston})-(\ref{HestonDyn}) with the equation  $ d I_t = v_t dt$, and solved the associated Fourier-transformed parabolic equation via the usual Feynman-Kac argument. As has to be the case, the two approaches coincide.}:
\begin{align}\label{hestonPhi}\Phi_{t_0}(z,w)
  =&\mathcal L^z_{ \Delta T}(z^2/2 + i z/2-iw ),\end{align}
 where  $ \mathcal L^z_{ \Delta T}$ indicates the transform with respect to $v_t^z$ which 
 is well-known
 analytically (e.g. Dufresne, 2001).  The case $T_t=t$ reverts back to the Black-Scholes model,  when (\ref{hestonPhi}) collapses to (\ref{BS}) with $\sigma^2=v_0$.

Other choices for $v_t$ are clearly possible, yielding different stochastic volatility models (the 3/2 model, GARCH, etc.). It is clear from the arguments above that, for an analytical  expression for $\Phi_{t_0}$ to exist it suffices that the Laplace transform of $T_t$ is known in closed form\footnote{See e.g. Lewis 2000, chapter 2, for the Laplace transform of the cited models.} and that $v_t$ belongs to a class of models that are stable under the Girsanov transformation.

\subsection{DTC jump diffusions}\label{DTCJD}

When the underlying L\'evy process is represented by a finite activity jump diffusion, operating a decoupled time change amounts to either introducing a stochastic volatility coefficient in the continuous Brownian part, or making the intensity of the compound Poisson process $X_t^d$ stochastic, or both. Models carrying this structure have been prominently discussed by D.S. Bates (1996) and H. Fang (2000).

\subsubsection{Stochastic volatility with jumps}\label{Bates}

The stochastic volatility model with jumps (SVJ) provides us with a first instance of  a decoupled time change not otherwise obtainable as an ordinary time change. The SVJ model is in fact a L\'evy decoupled time change with a time-changed continuous part and a time-homogeneous jump part. The dynamics for the asset price are given by the exponential jump diffusion: 
\begin{equation}dS_t=r S_{t^-} dt + \sqrt{v_t} S_{t^-} d W^1_t + S_{t^-} (\exp(J)-1) d N_t  - \kappa \lambda S_{t^-} dt;  \end{equation}
for some Brownian motion $W^1_t$, stochastic variance process $v_t$, Poisson process $N_t$ and jump size $J$ having compensator $\kappa$. The underlying DTC structure of the Bates model is given by $X_{T,U}=X^c_{T_t}+X^d_t$ with the characteristic triplet for $X_t$ being $(0, 1, \lambda f_J(x) dx)$ and $T_t$ taking the  form (\ref{T}).
By assuming as a jump distribution a normal random variable, and as a variance process the square-root equation:
\begin{equation}
 d v_t=\alpha(\theta-v_t)dt + \eta \sqrt{v_t}d W^2_t 
\end{equation}
we have the model by Bates (1996). For the discounted asset value to be a martingale, the parameters of the driving stochastic volatility and jump process must be subject to the requirements of both subsection \ref{stochvol} and subsection \ref{merton}.
It is straightforward to see that $\Phi_{t_0}(z,w)$ decomposes into:
\begin{equation}\Phi_{t_0}(z,w)=\Phi_{t_0}^c(z,w)\Phi^d_{t_0}(z,w),  \end{equation}
where $\Phi_{t_0}^c(z,w)$ and $\Phi_{t_0}^d(z,w)$
 are given respectively by (\ref{hestonPhi}) and (\ref{mertonPhi})-(\ref{mertonJumps}). Therefore:
\begin{align}\label{batesPhi}  \Phi_{t_0}(z,w)
= & \mathcal L^z_{\Delta T}(z^2/2 + i z/2-iw ) \exp(-(t-t_0)\lambda( iz \kappa - \phi_{J,J^2}(z,w)  + 1)).  
\end{align}

So far, we have encountered either exponential L\'{e}vy models, or exponentially-affine functions arising as solutions of a PDE problem. Here we have a mixture of the two: a time-homogeneous jump factor, modeled as a compound Poisson process,
 and a continuous diffusion factor, whose characteristic function solves a diffusion problem. The degenerate case $T_t=t$, yields a Merton jump diffusion with diffusion coefficient $\sqrt{v_0}$.

\subsubsection{Stochastic volatility with jumps and a stochastic jump rate}\label{Fang}

Another way of obtaining a DTC model is obtained by introducing a stochastic jump frequency into the jump diffusion of the log-price. A jump process with stochastic volatility  and stochastic jump rate (SVJSJ) has been suggested and empirically studied by  Fang (2000). For a time change $U_t$, we assume $N_t$ to be a pure jump process of finite activity such that conditionally on $U_t$, $N_t$ is distributed like a Poisson random variable of parameter $U_t$, and is  independent of every other involved process. We let $\lambda_t$ be another continuous stochastic process; with the remaining notation as in subsection \ref{Bates}, we define the asset price dynamics as follows:
\begin{equation}dS_{t}=r S_{t^-}dt+ \sqrt{v_t} S_{t^-} d W^1_t +  S_{t^-} (\exp(J)-1) d N_t - \kappa  \lambda_t S_{t^-} dt;  \end{equation}
 This model has a clear DTC L\'{e}vy structure $X_{T,U}$ given by $T_t$, $U_t$ as in (\ref{T}) and (\ref{U}) with $u_t=\lambda_t$, and the characteristic triplet $(0,1,f_J(x)dx)$. The model by Fang is obtained by setting:
\begin{equation}d v_t=\alpha(\theta-v_t)dt + \eta \sqrt{v_t}d W^2_t;
  \end{equation}
\begin{equation} d \lambda_t=\alpha_{\lambda}(\theta_{\lambda}-\lambda_t)dt+ \eta_{\lambda} \sqrt{\lambda_t}d W^3_t.
 \end{equation}
As usual we impose $\langle W^1_t, W^2 \rangle_t= \rho dt$; in contrast, the Brownian motion $W^3_t$ is assumed to be independent of all the other random variables. If both of the diffusion parameter sets obey Feller's condition and the density of $J$ decays sufficiently fast, $\tilde S_t$ is a martingale. Like in the Bates model, the jumps $J$ are normally distributed. The function $\Phi_{t_0}$ is then given by:
\begin{align}\label{fangPhi}\Phi_{t_0}(z,w)
= & \mathcal L^z_{\Delta T}(z^2/2 + iz/2 -iw)\mathcal L_{\Delta U}(iz\kappa - \phi_{J,J^2}(z,w)+1).
\end{align}
Again we recognize  that we can decompose $\Phi_{t_0}(z,w)=\Phi_{t_0}^c(z,w)\Phi^d_{t_0}(z,w)$, where 
 $\Phi_{t_0}^c$ is the leverage-neutral characteristic function of a Heston process of variance $v_t$, and  $\Phi_{t_0}^d$ that of a compound Poisson process time-changed with $U_t$, whose argument was computed in subsection \ref{merton}. The Laplace transforms of the integrated-square root processes arising from $v^z_t$ and $\lambda_t$ are known, and the leverage-neutral version $v_t^z$ of $v_t$ has been given in subsection \ref{stochvol}. Observe that there is no leverage effect in the jump part because of the assumptions on $W^3_t$. Finally, notice that the case $U_t=t$ reduces to the Bates model with a jump activity rate equal to $\lambda_0$.

\subsection{General exponentially-affine activity rate models}\label{affine}

 A general theory of affine models for the discounted asset dynamics has been laid out by   Duffie \emph{et al.} (2000), and Filipovi\'c  (2001), as well as others. We briefly illustrate how this ties in with  decouple time-changed processes. Suppose we have a  Markov process given by the stochastic differential equation: 
\begin{equation}\label{state} d Y_t= \mu(Y_t)dt + \sigma(Y_t) d W_t + d N_t\end{equation} 
 where $W_t$ is an $n$-dimensional Brownian motion, $N_t$ is an $n$-dimensional pure jump process of intensity  $\lambda(Y_t)$ and joint jump size distribution $F(x_1, \ldots, x_n)$ on $\mathbb R^n$. We fix a discount functional $R(x)= r_0 + r_1 \cdot x$, $(r_0, r_1) \in \mathbb R \times \mathbb R^n$ and assume for the coefficients the following linear structure:
\begin{equation}\label{affdiffparameters}
 \begin{array}{rl}  \mu(x)=m_0+m_1  x	, & (m_0, m_1) \in \mathbb R^n \times \mathcal M_n(\mathbb R) \\
  \sigma \cdot \sigma^T (x)=\Sigma_0+ \Sigma_1 x, & (\Sigma_0, \Sigma_1) \in  Sim_n(\mathbb R) \times Sim_n(\mathbb R) ^n   \\
   \lambda(x)=l_0+l_1 x,   &  (l_0, l_1) \in \mathbb R \times \mathbb R^n. \\
 \end{array}
\end{equation}

For some one-dimensional DTC process $X_{T,U}$, let $M_t$ be the change of measure martingale in (\ref{measurechange}) and assume $Y_t$ to be two-dimensional, so that the marginals of $Y_t$ represent the instantaneous activity rates $v_t$ and $u_t$. 

 The leverage-neutral characteristic function $\Phi_{t_0}$  can be recovered as follows. By taking the Ito differential of $\log M_t$ one sees that $M_t$ is itself a linear jump diffusion; we can thus define the three-dimensional augmented process $\tilde Y_t=(Y_t, \log M_t )$ having  some associated extended parameters $\tilde m_0, \tilde m_1,  \tilde \Sigma_0, \tilde \Sigma_1, \tilde l_0, \tilde l_1$ in (\ref{affdiffparameters}). Furthermore,  we can rewrite $M_t$ as:
\begin{equation}M_t=\exp(b \cdot \tilde Y_t) \end{equation}
where $b=(0,0,1)^T$. Now, according to the results of  Duffie \emph{et al.} (2000), appendix C,  under the measure $\mathbb Q=\mathbb Q(z,w)$ having Radon-Nikodym derivative $M_t$, we have:  
 \begin{equation}\label{extendedtransform}\Psi_{t_0}^\mathbb Q(u):=\mathbb E_{t_0} ^ \mathbb Q \left[\exp \left(-\int_{t_0}^t R(\tilde Y_s)ds \right) e^ {u \tilde Y_t}\right]=e^{-\alpha(t_0)  -\beta(t_0) \tilde Y_{t_0}}
\end{equation}
 for all $u$ for which  (\ref{extendedtransform}) is defined, and $\alpha(\tau)$, $\beta(\tau)$ following the Riccati system of ODEs\footnote{$(\beta(\tau)^T   \Sigma_1 \beta(\tau))^k:=\beta(\tau)^T   \Sigma_1^k \beta(\tau), \; k=1, \ldots ,n.$}:
  \begin{align}\label{odeQ1} \beta(\tau)' &= r_1 +(\tilde m_1^T+ \tilde \Sigma_1 b) \beta(\tau)  -   \frac{1}{2} \beta(\tau)^T   \tilde \Sigma_1 \beta(\tau)  - \tilde l_1 ( \mathcal L_F( \beta(\tau) +b )  - \mathcal L_F( b)  ) \\ \label{odeQ2} \alpha (\tau)' & = r_0  +(\tilde m_0+ \tilde \Sigma_0 b) \beta(\tau) -  \frac{1}{2} \beta(\tau)^T \tilde \Sigma_0 \beta(\tau)   - \tilde l_0 ( \mathcal L_F( \beta(\tau)+ b)  - \mathcal L_F( b)  )
\end{align}
for  $\tau \leq t $, with boundary conditions $\beta(t)=u$ and $\alpha(t)=0$. By choosing \begin{equation}\begin{array}{lr} r_0=0, &  r_1=(\zeta(z,w, \mu, \sigma, \theta_0), \xi(z,w, \nu(dx), \theta_0),0 ),\end{array}\end{equation}
one notices that: 
\begin{equation} \label{psiaffine}\Psi_{t_0}^\mathbb Q(0)= \Phi_{t_0}(z,w).
\end{equation}

The solvability of equations (\ref{odeQ1})-(\ref{odeQ2}) is discussed and characterized in Grasselli and Tebaldi (2007). What we have just shown is that the class of the exponentially-affine processes and that of the DTC L\'evy processes intersect in the class of the DTC processes whose instantaneous activity rates are given by affine jump diffusions of the form (\ref{state})-(\ref{affdiffparameters}).

We remark that $\tilde Y_t$ implicitly defines a price process $S_t$ through the instantaneous activity rates and the change of measure martingale $M_t$ accounting for the dependence structure between the time changes and the underlying L\'evy process. The augmented diffusion $\tilde Y_t$  is an \emph{exponentially-affine decoupled time change}; all the models reviewed so far fall under this category\footnote{A model that can be written in DTC form falling outside this intersection is the linear quadratic-affine model by Santa Clara and Yan, \cite{SantaClaraYan}. However, such model does not possess an analytically tractable transform to be used for pricing purposes.}.  Another example of a model that can be represented in this form is the ``double jump model'' of Duffie \emph{et al.} (2000), given by a jump diffusion with stationary jump intensity, whose stochastic volatility is itself a jump diffusion process having the same intensity as the stock.

\section{A multifactor DTC jump diffusion}\label{wish}

In this section we illustrate a theoretical model in the DTC framework admitting a closed formula for $\Phi_{t_0}$. The price evolution we consider has several attractive features: it is a DTC jump diffusion and therefore allows for the presence of a stochastic jump rate and a stochastic volatility; in addition, both of these processes are given through a multifactor specification. Regarding the correlation modeling, the dynamics we assume carry the usual linear correlation between the stochastic volatility and the Brownian motion driving the stock, as well as a dependence structure between the instantaneous rates of activity. 
Thus, the hypothesis of a market jump and continuous activity which are correlated with each other finds room in this model. 

 The case for multifactor volatility has been made by a number of authors. As pointed out by Bergomi 2005, a  volatility specification of this kind overcomes the inability of single factor models to fit the current market skew, while at the same time predicting the future evolution of implied volatilities consistently with the historical data, which is of particular relevance in the pricing of certain forward-starting derivatives such as the 
  cliquet option. Furthermore multifactor models make possible a long-term volatility specification that accounts for the slow decay of the autocorrelation function of the variance process (Gallant \emph{et al.} 2013), as opposed to single factor models, whose autocorrelation function typical decay is exponential.  

The price process we analyze links to a modern and currently very active strain of research, which makes use of the so-called \emph{Wishart process} for financial modeling purposes. The Wishart process is a matrix-valued affine process, studied foremostly by M.F. Bru (1991), that can be thought as a multivariate extension of the CIR process. It has been used to model the driving factors of term structures and price processes by, among the others, da Fonseca \emph{et al.} (2007, 2008), and  Gouri\'eroux and Sufana   (2003, 2010), among the others. 

For commuting matrices $Q$ and $M$ in  $\mathcal M_n(\mathbb R) $, with $Q$ invertible and $M$  negative definite (to capture mean-reversion), a Wishart process $\Sigma_t$ is defined by the following multi-dimensional SDE: 

\begin{equation} \label{Wishart} d \Sigma_t= \sqrt{\Sigma_t}  d B_t Q + Q^T d B_t^T \sqrt{\Sigma_t} + (M \Sigma_t +   \Sigma_t M^T + c Q^T Q)dt. \end{equation}

The Wishart process is thus a symmetric matrix-valued process. The matrix $M$ must satisfy the further constraint $c  \geq n-1$ for some $c>0$; $B_t$ is here an $n \times n$ \emph{matrix} of Brownian motions.

We can use $\Sigma_t$ to build a one-dimensional DTC jump diffusion model as follows. We choose $n=2$ and let $W_t$ be a two-dimensional Brownian motion such that $\langle  W^1, B^{1,1} \rangle_t=\langle W^2, B^{2,1} \rangle_t = \rho t$ for some correlation parameter $\rho$ and $W_t$ is independent of all the other entries of $B_t$.  Let  $N_t$ be a finite activity jump process  like in subsection \ref{Fang}, which we further assume it to be independent of both $W_t$ and $B_t$. As usual, the jump distribution $J$ is set to be independent of every other variable. Denoting by $\sigma_t$ the positive-definite matrix square root of $\Sigma_t$, we can define the risk-neutral dynamics of the log-price process $Y_t=\log(S_t/S_0)$ as:
\begin{equation}\label{WishPrice} \begin{array}{lr}   d Y_t=
(r - \Sigma^{1,1}_t/2 - \Sigma_t^{2,2} \kappa  )dt +  \sigma^{1,1}_t d  W^1_t   + \sigma^{1,2}_t d W^2_t + J d N_{t}, \ & Y_0=0 \end{array}
\end{equation}
where $\kappa$ equals $\phi_J(-i)-1$. The process $S_t$ can be seen to be a local martingale of the form (\ref{AssetDynamics}) by assuming the time changes in proposition \ref{prop1}  to be like those in equations (\ref{T}) and (\ref{U}) and letting:
\begin{equation}\label{DTCwish}\begin{array}{ccccc} 
  \displaystyle{d  X^c_{t} =\frac{\sigma_t^{1,1}}{\sqrt{\Sigma^{1,1}_t}}d W_t^1 +\frac{\sigma_t^{1,2}}{\sqrt{\Sigma^{1,1}_t}}d W_t^2} ,  &
   X^d_t  = \sum ^{M_t}_{i=0}  J, &
   v_{t}  =  \Sigma_t^{1,1}, &
   u_t   =  \Sigma_t^{2,2}, &
    \theta_0 =  (-i, -i). 
   \end{array}
\end{equation}
where $M_t$ is a Poisson process of intensity 1. Multifactoriality is reflected in the fact that 
 even though each activity rate is specified by a single factor, the correlation  between them involves all the three components of $\Sigma_t$. Indeed, let $w^j_t$ be the scalar Brownian motion driving $\Sigma_t^{j,j}$: it can be proved that
 \begin{equation}\label{activitycorr}
 d \langle w^1, w^2 \rangle_t= \frac{\Sigma_t^{1,2}(Q^{1,1} Q^{1,2} + Q^{2,1} Q^{2,2})}{\sqrt{\Sigma_t^{1,1} ((Q^{1,1})^2+ (Q^{2,1})^2) } \sqrt{\Sigma_t^{2,2}((Q^{1,2})^2+ (Q^{2,2})^2) } }dt.
 \end{equation}
 Observe that this correlation is \emph{stochastic}. The correlation between $Y_t$ and its instantaneous variance  $\Sigma^{1,1}_t$ is instead determined by the interplay between $\rho$ and $Q$; we have: 
\begin{align}\label{wishcorr} d \langle w^1, X^c \rangle_t = \frac{\rho Q^{1,1}}{\sqrt{(Q^{1,1})^2 + (Q^{2,1})^2}}.
\end{align}
 By applying the Girsanov's transformation, we see that the $\mathbb Q(z,w)$-dynamics of (\ref{Wishart}) are given by the complex-valued Wishart process:
\begin{equation} \label{Wishartz} d \Sigma^z_t= \sqrt{\Sigma^z_t} d B_t Q + Q^T d B_t^T \sqrt{\Sigma^z_t} + (M^z  \Sigma^z_t +  \Sigma^z_t (M^z)^T + cQQ^T)dt 
\end{equation}
where
\begin{equation} \begin{array}{lr}M^z=M + i z Q^T R , & R=  \left( \begin{matrix} \rho & 0 \\ 0 & 0
\end{matrix} \right), 
\end{array}
\end{equation}
whence:
\begin{align}\label{WishartPhi} \Phi_{t_0}(z,w) 
= &    \mathcal L^z_{\Delta T, \Delta U} ( z^2/2 + iz/2-iw , iz\kappa - \phi_{J,J^2}(z,w)+1).
\end{align}

Notably, the Laplace transform  $\mathcal L_{\Delta T,  \Delta U}(\cdot)$ for $v_t$ and $u_t$ as in (\ref{DTCwish}) can be derived in closed form (see the appendix), since it is a particular case of some well-studied transforms of the Wishart process. In the spirit of the previous section, we remark that when all the Wishart matrices are diagonal, an application of the L\'evy theorem shows that this model reduces to a SVJSJ for some appropriate parameter choice (see section \ref{numericalwish} below). 

 It is therefore possible to price, and find price sensitivities of, joint price/volatility contingent claims on an asset whose log-price process follows $Y_t$. The model just presented is a particular DTC jump diffusion featuring not only the usual leverage effect between the underlying jump diffusion and the continuous/jump market activity, given by (\ref{wishcorr}), but also a correlation structure between the rates of activities themselves, as shown by equation (\ref{activitycorr}). For practical financial modeling purposes, we shall see in section \ref{numericalwish} below that this relationship positively impacts the ability of a jump model of capturing the volatility skew of a traded asset.

 This asset pricing model provides an example of how non-trivial DTC modeling (i.e. achieved by using \emph{dependent} time changes) might work in practice. As a general approach, one could start from a multivariate stochastic process whose integrated marginals have a known joint Laplace transform, and use these as time changes for the continuous and discontinuous parts of some given L\'evy process. The underlying L\'evy triplet will only appear as an argument of such a transform, and the characteristic function of the process is then completely determined up to a measure change.  This and similar models are currently the subject of further research.


\section{Numerical testing}\label{numerical}

\subsection{Implementation of the pricing formula}

For validation purposes, we numerically implemented equation (\ref{Lewis}) in MATHEMATICA\textsuperscript{\textregistered} for various models and payoffs, and compared the analytical prices so obtained to a MATLAB\textsuperscript{\textregistered} simulation following an Euler scheme. The results confirm the consistence of the pricing formula with the risk-neutral valuation theory.

We analyzed three different contingent claims: one on $S_t$, one on $TV_t$, and one joint derivative on $S_t$ and $ TV_t$. Namely, we accounted for three different kinds of options: a vanilla call option, a call option on the realized volatility, and a call TVO.

For a plain call option of maturity $t$ and strike $K$, the function $F$ and its Fourier transform $\hat F$ to be used in (\ref{Lewis}) are:
\begin{equation}\begin{array}{lr} \displaystyle{ F(z)=(e^z-K)^+}, & \displaystyle{\hat F(z)= \frac{ K^{1+i z}}{(i z- z^2)}};
\end{array}\end{equation}  
the function $\hat F$ exists and is analytic for  Im$(z) >1$.

A possible volatility investment is to write a call option using as an underlying the total realized volatility $\sqrt{TV_t}$ of an asset, or to buy a call option directly on a volatility index such as the VIX. Hence, we would like to price the contingent claim paying $(\sqrt{TV_t}-Q)^+$ at time $t$ for some strike realized volatility level $Q$. In our equation we would then need to take: 
\begin{equation}\begin{array}{lr} \displaystyle{F(w)=(\sqrt{w}-Q)^+,} & \displaystyle{\hat F(w)=\frac{\sqrt{\pi} \mbox{ Erfc}(Q \sqrt{-i  w}) }{2 (-i w)^{3/2}};}  \end{array}
\end{equation}  
the Fourier transform here is well-defined and holomorphic in Im$(w)>0$.  
  
 The target volatility option mentioned in the introduction is a natural candidate for testing mixed-claim structures, being an instance of a currently traded joint asset/volatility derivative. The payoff function $F$ and the Fourier transform for a call TVO of strike $H$, maturing at $t$  with target volatility $\target$ are:   
 \begin{equation}\begin{array}{rl} \displaystyle{ F(z, w)=\target \sqrt{\frac{t}{w}} (e^z-H)^+,} & \displaystyle{\hat F(z,w) =  \target(1+i) \sqrt{ \frac{\pi t}{2 w}}  \frac{ H^{1+i z}}{(i z- z^2)}}. \end{array}
\end{equation}
Observe that, unlike the previous contracts, the payoff $F$ of a TVO shows explicit dependence on the expiry $t$. The domain of holomorphy of $\hat F$ is the strip $\Sigma_F=\{ (z,w) \in \mathbb C^2, \mbox{ Im}(z)>1, \mbox{ Im}(w)>0 \}$.

\medskip

We numerically tested these derivatives using five different stochastic models for the underlying asset processes: namely, the Black-Scholes, Heston, Merton, Bates and Fang models. All the prices have been produced with a single implementation of (\ref{Lewis}) with $\Phi_{t_0}$ given by (\ref{fangPhi}). All we had to do is changing/voiding the relevant parameters,  and replacing the module for $\hat F$ whenever we switched payoff. The parameter estimates  have been taken from Fang's fitting of the S\&P 500 index, and are illustrated in table 1. Tables 2 to 6 summarize the result obtained for five different sets of observable market conditions $(r, t_0, S_{t_0},  TV_{t_0})$ and contract parameters $t, K,Q,H, \target$. For each given $t_0$, the maturity $t>t_0$ is  the same for all the three options considered; a TVO is always compared to a vanilla call having same strike, and the target volatility is set to be constant across all the data sets.   

We simulated 100.000 paths of step size $(t-t_0)/1000$. The figures show a good overall match between the analytical value (AV) and the Monte Carlo value (MC); the relative error $|$AV$-$MC$|$/MC is shown in parentheses.  For the call option on the volatility in some cases we almost attain four-digit precision. On the other hand, for some models and data sets the integrands for the TVO valuation remain highly oscillatory around the maximum integration range; when this occurs, a certain loss of accuracy is observed.  

\subsection{Leverage sensitivity of the model of section \ref{wish}}\label{numericalwish}

In this section we try to provide further financial motivation for the model with correlated stochastic volatility and jump rate (CSVJA) of section \ref{wish} through a numerical exercise. Specifically, in what follows we analyze the sensitivity of the volatility skew of the Bates, Fang and CSVJA models with respect to variations in the ``leverage parameter'' $\rho$.  We define in this section $\rho$ as being the constant value giving the instantaneous correlation between the Brownian components of the stochastic variance and that of the log-returns of the accounted models. 

 Preliminarily, let us assume the following form for the matrices in the equations of section 7:
\begin{equation}\label{matrices}
\begin{array}{cccc}
M= \left( \begin{matrix}
 M_1 & 0 \\ 0  & M_2  \
\end{matrix} \right),  & \Sigma_0= \left( \begin{matrix}
\Sigma_1 & 0 \\ 0 & \Sigma_2 
\end{matrix} \right), & R= \left( \begin{matrix} R_1 & 0 \\ 0 & 0
\end{matrix} \right),
Q= \left( \begin{matrix}  Q_1 & Q_0 \\ Q_0 & Q_2
\end{matrix} \right).
\end{array} 
\end{equation}
The jumps are assumed to be normally distributed with mean $\mu$ and variance $\delta^2$. As already mentioned in section 6, an application of the L\'evy theorem shows that when $Q_0=0$ the parametrization above  is equivalent to a Fang model with:
\begin{equation}
\begin{array}{cccccc}
\alpha=-2M_1, & \displaystyle{\theta=-\frac{b Q_{11}^2}{2 M_1}}, & \eta=2 Q_{11}, & v_0=\Sigma_1; \\   \\ \alpha_\lambda=-2M_2, & \displaystyle{\theta_\lambda=-\frac{b Q_{22}^2}{2 M_2}}, & \eta_\lambda=2 Q_{22}, & \lambda_0=\Sigma_2, & \rho=R_1.
\end{array}
\end{equation}
Clearly, by further assuming $M_2=Q_2=0$ one recovers the Bates model. The example in (\ref{matrices}) is then ``minimal'' in the sense that it represents the slightest possible modification of certain known models attaining a true CSVJA specification. 
  We set the parameters in  (\ref{matrices}) as follows:
\begin{equation}\label{parametersmatrix}
\begin{array}{ccccccc}
M_1=M_2=-0.33, &  Q_1=Q_2=0.25, & Q_0=0.15, & \Sigma_1=\Sigma_2=0.01, & \kappa=-0.2, & \delta=0.3, & c=3. 
\end{array}
\end{equation}
 The market parameters have been chosen as $S_{t_0}=100$ and $r=0.03$. According to the remarks above, by voiding the relevant matrix entries, we can use (\ref{parametersmatrix}) to specify one instantiation each of the Bates, Fang and CSVJA models.   

In figures \ref{fig1} and \ref{fig2} we show the 3-month volatility skew extracted from the call option prices generated by (\ref{LewisEq}), respectively for the Fang and Bates model. The two curves correspond to a value of $\rho$ given respectively by $0$ and $-0.6$. As we can see, the skew is only marginally affected by the variation of the leverage parameter.  We emphasize that, consistently with the standard representations given in section 6, a direct check confirms that we have here no correlation between the activity rates, that is, if $Q_0=0$ then (\ref{activitycorr}) vanishes.

The lack of sensitivity of the skew with respect to $\rho$ in jump models is a well understood fact. The reason is that for such a class of models the short term dynamics of the surface are generally handled by the jump parameters which unlike the leverage value, do not retain a clear economic interpretation. This is normally regarded as a shortcoming of the jump models, since an improvement in the short term smile fit is achieved at a cost of a loss of sensitivity with respect to the calibrated parameters, which can be problematic for e.g. skew hedging purposes. For a full discussion, see da Fonseca and Grasselli (2011). 

Interestingly, if we perform the same analysis for the full CSVJA parametrization (figure \ref{fig3}) 
 we instead notice a considerable variation of the volatility skew when leverage is introduced. In particular, we see that when $\rho$, which in this case is given by equation (\ref{wishcorr}), is nonzero, the short term smile is much more negatively skewed than in the uncorrelated cases. Since we are exactly in the same situation as in the Fang model in terms of marginal distributions of the driving factors and values of $\rho$, such an increased sensitivity of the skew  can only be due to the  correlation between the activity rates established when letting $Q_0 \neq 0$. The effect of this ``second'' correlation on the skew can be intuitively explained as follows. Because of the leverage effect, as prices go down the volatility spikes up, generating negative skewness in the returns distribution. But now the volatility is correlated with the jump activity; in particular, we see from (\ref{activitycorr}) that this is most likely going to be a \emph{positive} correlation. The reason for this is that the sign of (\ref{activitycorr}) depends only on the sign of $\Sigma^{1,2}_t$, which is a mean reverting process with a positive mean reversion level. Therefore, when the volatility increases, the jump intensity is likely to increases as well, and hence so does the probability of observing a (negative on average) jump. The latter contributes to the negative skewness of the asset returns and thus reinforces the negative skew of the volatility smile. 

The test conducted reveals a very useful property of a DTC L\'evy jump model of CSVJA type. As opposed to the standard jump models with independent activities, the underlying dynamics introduced are able to adequately match the convexity of a steep volatility skew without having to surrender the overall control on the surface of the correlation parameter $\rho$.  Unlike the traditional jump models, changes in the short term part of the surface can be achieved not only by a change in the distribution of the jump part, but also by varying a correlation parameter, much like what happens in a purely diffusive model. The ability of the surface to promptly respond to variations of financially meaningful variables is a property greatly appreciated by the practitioners; in this respect the CSVJA model may improve the jump asset modeling literature. An empirical study on this model is matter of ongoing research.

\section{Conclusions}

In this paper we have suggested a theoretical pricing framework that can easily be made to represent popular settings, but whose full model and payoff generalities were not possible by using the previous theory. We achieved this by introducing the concept of decoupled time change and by considering payoffs on an asset and its accrued volatility as the default target claims to be priced.

 DTC processes provide  a common time-changed representation for many models from the extant literature, and help to capture possible dependence relationships between the continuous and the jump market activities. We obtained martingale relations for stochastic exponentials of DTC L\'{e}vy processes, based on which we defined an asset price's dynamics. We then linked the joint characteristic function of the log-price dynamics and the quadratic variation to the joint Laplace transform of the time changes. As a by-product, we extended the measure change technique of Carr and Wu (2004) to the class of DTC L\'evy processes. In the DTC setup, we rigorously posed and solved the valuation problem of a derivative paying off on an asset  $S_t$ and its realized volatility,  
 by  means of an inverse-Fourier integral relation that extends previously known formulae. 

Several stochastic models and contingent claims have been analyzed. In all the accounted cases we outlined the underlying DTC structure and found the leverage-neutral characteristic function. In particular, the SVJ and SVJSJ models were shown to have their own time-changed L\'{e}vy structure. Furthermore, we have introduced a novel DTC L\'evy theoretical model which illustrates how equity modeling could benefit from the idea of decoupled time changes.   

For numerical comparison and validation purposes, we focused on specific instances from the three payoff classes allowed by our equation: plain vanilla claims, volatility claims, and joint asset/volatility claims. The results confirm the validity of our method.  From a computational standpoint, a single software implementation can output prices for several different combinations of models and payoffs. 
Finally, we have presented some initial evidence that a model with correlation between the activity rates potentially allows a better management of the volatility skew compared to other jump models.
 


\section*{Appendix: proofs}

We begin by recalling some basic definitions from the semimartingale representation theory; in particular, we refer to Jacod and Shiryaev (1987), chapters 2 and 3, and Jacod (1979), chapitre X. 

\bigskip

 We define the \emph{Dol\'{e}ans-Dade exponential} of an $n$-dimensional semimartingale $X_t$ starting at 0 as:
\begin{equation} \mathcal E(X_t)=e^{X_t- \langle  X^{c}_t \rangle/2} \prod_{s \leq t}(1+ \Delta X_s)e^{- \Delta X_s}
\end{equation}
where $X^{c}_t$ denotes the continuous part of $X_t$ and the infinite product converges uniformly. This is known to be the solution of the SDE $d Y_t=Y_{t^-}dX_t$, $Y_0=1$. 

\bigskip

Let $\epsilon(x)$ be a truncation function and $(\alpha_t, \beta_t, \rho(dt \times dx))$ be a triplet of predictable processes that are well-behaved in the sense of  Jacod and Shiryaev (1987), chapter 2, equations (2.12)-(2.14). For $\theta \in \mathbb C^n$, associate with  $(\alpha_t, \beta_t, \rho(dt \times dx))$ the following complex-valued functional:
\begin{equation}\label{cumulant} \Psi_t(\theta)= i \theta^T \alpha_t - \theta^T \beta_t  \theta/2+ \int_0^t\int_{\mathbb R^n} (e^{i \theta^T x}-1-i \theta^T x \epsilon(x) )\rho(ds \times dx). 
\end{equation}
This functional is well-defined on:
\begin{equation}\label{domain} \mathcal D= \left \{\theta \in \mathbb C^n \mbox{ such that } \int_0^t \int_{\mathbb R^n}e^{i \theta^T x} \epsilon(x) \rho (ds \times dx) < + \infty   \mbox{ almost surely} \right \} \end{equation}
and because of the assumptions made it is also predictable and of finite variation.

\bigskip

Let $X_t$ be an $n$-dimensional semimartingale.   The \emph{local characteristics} of $X_t$  are the unique predictable processes $(\alpha_t, \beta_t, \nu(dt \times dx))$ as above, such that $\mathcal E(\Psi_t(\theta)) \neq 0$ and  $\exp(i \theta^T X_t)/\mathcal E(\Psi_t(\theta))$ 
is a local martingale for all $\theta \in \mathcal D$. The process $\Psi^X_t(\theta)$ in (\ref{cumulant}) arising from the local characteristics of $X_t$ is called the \emph{cumulant process} of $X_t$, and it is independent of the choice of $\epsilon(x)$. It is clear that the local characteristics of a L\'evy process $X_t$ of L\'evy triplet $(\mu, \Sigma, \nu)$ are $(\mu t, \Sigma t, \nu dt  )$.

\bigskip

If $\mathcal B$ is a Borel space, the time change of a random measure  $\rho( d t \times dx)$ on the product measure space $\Omega \times \mathcal B(\mathbb R_+ \times \mathbb R^n)$ according to some time change $T_t$, is the random measure:
\begin{equation}
\rho( d T_t \times dx)(\omega, [0,t) \times B)=\rho( dt \times dx) (\omega, [0, T_t(\omega)) \times B)
\end{equation}
for $\omega \in \Omega$, $t \geq 0$ and all sets $B \in \mathcal B(\mathbb R^n)$. A random measure $\rho(dt \times dx)$ is $T_t$-\emph{adapted}  if for all $t, \omega$ and $B$ holds $\rho(dt \times dx)((T_{t^-}, T_t], \omega, B)=0$. This is equivalent to say that for each measurable random function $W$, the integral of $W$ with respect to $\rho$ is $T_t$-continuous (see Jacod 1979, chapitre X); conversely, if $X_t$ is a pure jump process that is $T_t$-continuous, then its associated jump measure $\rho(dt \times dx)$ is $T_t$-adapted (Kallsen and Shiryaev 2002, proof of lemma 2.7).

\bigskip

A semimartingale $X_t$ is said to be \emph{quasi-left-continuous} if its local characteristic $\nu$ is such that $\nu( dt \times dx )(\omega,  \{t \} \times  B)=0$ for all $t \geq 0$, Borel sets $B$ in $\mathbb R^n$, and $\omega \in \Omega$. Essentially, quasi-left-continuity means that the discontinuities of the process cannot occur at fixed times.

\bigskip

The following theorem clarifies the importance of continuity/adaptedness under time changing, i.e. that stochastic integration and integration with respect to a random measure ``commute'' with the time changing operation.

\begin{ThmA}\label{thma}  Let $T_t$ be a time change with respect to some filtration $\mathcal F_t$.
\begin{itemize}
\item[(i)]
Let $X_t$ be a $T_t$-continuous semimartingale. For all $\mathcal F_t$-predictable integrands $H_t$, we have that $H_{T_t}$ is $\mathcal F_{T_t}$-predictable, and:
\begin{equation} \int_{0}^{T_t}H_{s}dX_s=\int_{0}^{t} H_{T_{s^-}}dX_{T_s};\end{equation}
 \item[(ii)] Let $\rho(dt \times dx)$ be a $T_t$-adapted random measure on  
 $\Omega \times \mathcal B(\mathbb R_+ \times \mathbb R^n)$. For all measurable random functions $W(t, \omega, x)$ and $\omega \in \Omega$ it is:
\begin{equation} \int_{0}^{T_t} \int_{\mathbb R^n}W(s, \omega,x) \rho(ds \times dx)(\omega)=\int_{0}^{t} \int_{\mathbb R^n} W(T_{s^-}(\omega), \omega, x) \rho(dT_{s} \times dx)(\omega).\end{equation}
\end{itemize}
\end{ThmA}

\begin{proof}
See Jacod (1979),  th\'{e}or\`{e}me 10.19, (a), for part \emph{(i)}, and th\'{e}or\`{e}me 10.27, (a), for part \emph{(ii)}.
\qed
\end{proof}

In particular, from part \emph{(ii)} of theorem A follows that if $X_t$ is a pure jump proces with associated jump measure $\rho(dt \times dx)$ adapted to some time change $T_t$, then the time-changed process $X_{T_t}$ has associated jump measure $\rho(dT_t \times dx)$. 

\bigskip
 
It is essentially a consequence of theorem A that under the assumption of  continuity with respect to $T_t$, the local characteristics of a time-changed semimartingale are well-behaved, in the sense of the next theorem.

\begin{ThmB}\label{thmb}
Let $X_t$ be a semimartingale having local characteristics $(\alpha_t, \beta_t, \rho(dx \times dt))$ and cumulant process $\Psi^X_t(\theta)$ with domain $\mathcal D$, and let $T_t$ be a time change such that $X_t$ is $T_t$-continuous. Then the time-changed semimartingale $Y_t=X_{T_t}$ has local characteristics $(\alpha_{T_t}, \beta_{T_t}, \rho( d T_t \times dx))$ and  the cumulant process $\Psi_t^Y(\theta)$ equals $\Psi^X_{T_t}(\theta)$, for all $\theta \in \mathcal D$. 
\end{ThmB}
\begin{proof} See Kallsen and Shiryaev (2002), lemma 2.7.  \qed \end{proof}

\bigskip

\begin{proof}[Proof of proposition \ref{prop1}]
Let $(\mu, \Sigma,0)$ and $(0,0, \nu)$ be the L\'evy triplets of $X^1_t$ and $X^2_t$. Because of the $T^1_t$ and $T_t^2$-continuity assumption, we can apply theorem B and we immediately see that the local characteristics of $X^1_{T^1_t}$  and  $X^2_{T^2_t}$ are respectively $(T^1_t \mu, T^1_t \Sigma  , 0 )$ and $(0,0, dT^2_t \nu)$. By a result on the linear transformation of semimartingales, such two sets of local characteristics are additive (in Eberlein \emph{et al.} (2009), proposition 2.4, take $U$ to be the juxtaposition of two $n \times n$ identity blocks and $H=(X^1_{T^1_t}$  $X^2_{T^2_t})^T$ ), so that  $X_{T_t}$ has local characteristics\footnote{The process $X_{T_t}$ is a particular instance of an \emph{Ito semimartingale}: see Jacod and Protter (2003).}
 $(T^1_t \mu, T^1_t \Sigma  , dT^2_t \nu  )$

  Let $\Psi_t(\theta)$ be the cumulant process of $X_{T_t}$; by definition  the exponential $\mathcal E(\Psi_t(\theta))$ is well-defined if and only if $\theta \in \Theta$.
 But now the fact that $T^1_t$ and $T_t^2$ are continuous implies that $X_{T_t}$ is quasi-left-continuous (Jacod and Shiryaev 1987, chapter 2, proposition 2.9),
 that in turn is sufficient for $\Psi_t(\theta)$ to be continuous (Jacod and Shiryaev 1987, chapter 3, theorem 7.4). Therefore, since $\Psi_t$ is of finite variation, we have that 
$\mathcal E(\Psi_t(\theta))=\exp(\Psi_t(\theta))$;
in particular, this means that $\mathcal E(\Psi_t(\theta))$ never vanishes. By definition of the local characteristics, we then have that $M_t(\theta, X_t, T_t)$ is a local martingale for all $\theta \in \Theta$, 
 and thus it  is a martingale if and only if $\theta \in \Theta_0$. \qed
\end{proof}

\bigskip

\begin{proof}[Proof of proposition \ref{linearcommutativity}] An immediate consequence of theorem B is that, under the present assumptions, the class of continuous and pure jump martingales are closed under time changing, so that orthogonality follows. Therefore:
\begin{equation}\langle X_{T,U} \rangle_t= \langle X^c_{T} \rangle_t + \langle X^d_{U} \rangle_t.  \end{equation}
 The equation $\langle X^c_T \rangle_t=\Sigma T_t=\langle X^c \rangle_{T_t}$ can be established by the application of Dubins and Schwarz theorem. Regarding the discontinuous part, we notice that if $\rho( d t \times dx)$ is the jump measure associated to $X^d_t$  we have that
 $\rho$ is $U_t$-adapted because $X_t^d$ is $U_t$-continuous. Hence, the application of theorem A, part \emph{(ii)}, yields: \begin{equation}\langle X^d \rangle_{U_t}=\sum_{t < U_t}(\Delta X_t)^2 = \int_0^{U_t} x^2 \rho(ds \times dx)= \int_0^{t} x^2 \rho(dU_s \times dx)=\langle X^d_U \rangle_t. \end{equation} \qed 
\end{proof}

\begin{proof}[Counterexample to proposition \ref{linearcommutativity}] Let $X_t^c$ be a standard Brownian motion, and let $T_t$ be an inverse Gaussian subordinator with parameters $\alpha>0$ and $1$, independent of $X_t^c$.  The process $X_{T_t}^c$ is a \emph{normal inverse Gaussian process} of parameters  $(\alpha, 0, 0, 1)$ and is a pure jump process  (Barndorff-Nielsen 1997). Therefore by letting $X^d_t=X_{T_t}^c$ and $U_t=t$ we have $X^c_{T_t}=X_{U_t}^d$ so that orthogonality does not hold; moreover
$\langle X_{T,U}\rangle_t=2 \langle X^d \rangle_t$ while the left hand side of (\ref{ortocomm}) equals $\langle T \rangle_t + \langle X^d \rangle_t$.
\qed
\end{proof}

\bigskip

\begin{proof}[Proof of proposition \ref{DTCchar}]
 Since $T_t$ and $U_t$ are of finite variation, the total realized variance of an asset as in (\ref{AssetDynamics})  satisfies $TV_t= -\theta^2_0 \langle X_{T,U}\rangle_t$, so that  by proposition \ref{linearcommutativity} we have:
\begin{equation}\label{TV}
TV_t=  -\theta^2_0( \sigma^2{T_t} + \langle X^d \rangle _{U_t}).
\end{equation}
The application of proposition \ref{prop1} to  $C_t + D_t$ guarantees that the process in (\ref{measurechange})  is a martingale for all $z,w \in \mathbb C$ such that $(iz\theta_0,i w \theta_0) \in \Theta_0$. By using  relation (\ref{TV}) and operating the change of measure entailed by (\ref{measurechange}) we have: 
\begin{align}\label{cf1} \Phi& _{t_0} (z,w)=  \e_{t_0}[ \exp(iz\log( \tilde S_t/S_{t_0} ) + i w  \,  (TV_t- TV_{t_0}  )] \nonumber \\  = & \e_{t_0}  [  \exp(iz (i \theta_0 (\Delta X^c_{T_t} + \Delta X^d_{U_t}) - \Delta T_t \psi^c_X(\theta_0) - \Delta U_t \psi^d_X(\theta_0)) - i w \theta_0^2 ( \sigma^2 { \Delta T_t} + \Delta \langle X^d \rangle_{U_t}) )]  \nonumber \\ = & 
 \e_{t_0}  [ \exp(  i(iz \theta_0, iw \theta_0) \cdot ( \Delta C_{T_t} + \Delta  D_{U_t}) - \Delta T_t (iz\psi^c_X(\theta_0)+i  w \theta_0^2 \sigma^2)  -  \Delta U_t i z \psi^d_X(\theta_0) )]  \nonumber \\   
 = & \e ^ {\mathbb Q}_{t_0}[   \exp ( -\Delta T_t ( \theta_0 \mu  (z - i z)  - \theta^2_0 \sigma^2 (  z^2 + i z     - 2 i w )/2) \nonumber- \Delta U_t(i z \psi_X^d(\theta_0)- \psi_D(iz\theta_0, iw \theta_0)) )].
  \end{align}
 To fully characterize $\Phi_{t_0}$ all that is left is expressing $\psi_D$ in terms of $\nu$. Since 
\begin{equation}\label{Dimplicit} \psi_D(z,w)=\log \e \left[\exp \left(\sum_{s<t}i z \Delta X^d_s +  i w (\Delta X_s^d)^2 \right)\right],\end{equation}
we have that:
\begin{equation}\label{D} \psi_D(z,w)= \int_{\mathbb R}(e^{iz x + iw x^2}-1-i(z x+  w x^2) \I_{|x| \leq 1})\nu(dx) \end{equation}
which completes the proof.  \qed
\end{proof}

\bigskip

\begin{proof}[Proof of proposition \ref{LewisThm}] We follow the proof Lewis (2001), theorem 3.2, lemma 3.3 and theorem 3.4. By writing the expectation as an inverse-Fourier integral (which can be done by the assumptions on $F$ and because $\Phi_{t_0}$ is a characteristic function) and passing the expectation under the integration sign we have:
\begin{align}\e_{t_0} &[e^{-r(t-t_0)}F(Y_t, \VolY)]=\e_{t_0} \left[\frac{e^{-r(t-t_0)}  }{4 \pi^2} \int_{i k_1-\infty}^{i k_1 +\infty}  \int_{i k_2-\infty}^{i k_2+\infty}  S_t^{-i z} e^{- i w \VolY } \hat F(z,w)dz dw \right]  \nonumber \\ & =\frac{ e^{-r(t-t_0)} }{4 \pi^2} \int_{i k_1-\infty}^{i k_1 +\infty}  \int_{i k_2-\infty}^{i k_2+\infty} e^{-i w \langle Y \rangle _{t_0}} S_{t_0}^{-i z} e^{-r(t-t_0)iz} \Phi_{t_0}(-z, -w) \hat F(z,w)dz dw.  
\end{align}

All that remains to be proven is that Fubini's theorem application is justified. Let  $N_t=\log M_t(\theta_0, X_t, (T_t, U_t))$ be the discounted, normalized log-price; define the probability transition densities $p_t(x,y)=\mathbb P(N_t< x, \langle N \rangle_t <y| \; t_0, N_{t_0}, \langle N \rangle_{t_0})\I_{\{x \in \mathbb R, y \geq \langle N \rangle_{t_0}\}}$,  and let $\hat p_t(z,w)$ be their characteristic functions.
For all $(z,w) \in L_{k_1,k_2}$ we have:
  \begin{align} \int_{i k_1-\infty}^{i k_1 +\infty}  \int_{i k_2-\infty}^{i k_2+\infty} & \Big|e^{-i w \langle Y \rangle _{t_0}} S_{t_0}^{-i z} e^{-r(t-t_0)i z} \Phi_{t_0}(-z, -w)\Big| \hat F(z,w)dz dw   \nonumber \\ \label{parseval} = & \int_{i k_1-\infty}^{i k_1 +\infty}   \int_{i k_2-\infty}^{i k_2+\infty}   \hat p_t(-z,-w)  \hat F(z,w)dz dw  \nonumber  \\=& \int_{\mathbb R^2}    \hat p_t(-z+i k_1,-w+i k_2)   \hat F(z+i k_1,w+ i k_2)dz dw. 
  \end{align}

For $x \in \mathbb R$, $y \geq 0$, set $f(x,y)=e^{ -k_1 x -  k_2 y }F(x,y)$  $g(x,y)=e^{k_1 x + k_2 y}p_t(x,y)$. We see that the integrand in the right-hand side of (\ref{parseval}) equals   $ \hat g^*(z,w)\hat {f}(z,w)$. 
 But now $f$ is $ L^1(dx \times dy)$ because $F$ is Fourier-integrable in $\Sigma_F$ (for $(z,w) \in \Sigma_F$ take $\mbox{Re}(z)=\mbox{Re}(w)=0$); similarly, $\hat g^*$ is $L^1(dz \times dw)$ because of the $L^1$ assumption on $\Phi_{t_0}$.  
 Therefore, the application of Parseval's formula 
 yields:
\begin{align}  \int_{-\infty}^{ +\infty}  \int_{-\infty}^{+\infty} &    \hat p_t(-z +i k_1,-w+i k_2) \hat F(z+ik_1,w+ik_2)dz dw  \nonumber \\ = 4 \pi^2 \int_{-\infty}^{ +\infty} & \int_{-\infty}^{+\infty}    p_t(x,y)F(x,y)dx dy  = 4 \pi ^2 \e_{t_0}[ F(N_t, \langle N \rangle_t)] < +\infty,
\end{align}
since $F \in  L^1_{t_0}(N_t, \langle N  \rangle_t)$. 
\bigskip
\end{proof}

\begin{proof}[Proof of the equations of section \ref{wish}]

We can endow $Y_t$ with a correlation structure as follows. Let $Z_t$ be a two-dimensional matrix Brownian motion independent of $W_t$. The matrix process:
\begin{equation}  B_t= \left( \begin{matrix}     \rho  W^1_t  +   \sqrt{ 1 - \rho^2} Z^{1,1}_t & Z_t^{1,2}\\ \rho  W^2_t  +   \sqrt{ 1 - \rho^2} Z^{2,1}_t  & Z_t^{2,2} \end{matrix}   \right)
\end{equation}
is also a matrix Brownian motion enjoying the property that $\langle  W^j, B^{j,1} \rangle_t = \rho t$ and $ W_t$ is independent of $B_t^{j,2}$  for $j=1,2$. Since $\Sigma_t^{i,i}=(\sigma_t^{i,i})^2+(\sigma_t^{1,2})^2$, we have that $X_t^c$ is indeed a Brownian motion and the activity rates are connected through the element $\sigma_t^{1,2}$. 

To verify equations (\ref{activitycorr}) and (\ref{wishcorr}), observe that for $j=1,2$ there exist some bounded variation processes $F^j_t$ such that
\begin{equation}
 d \Sigma_t^{j,j}= F_t^j dt + 2 \sigma_t^{1,j}(Q^{1,j}d B_t^{1,1} + Q^{2,j} d B_t^{1,2})+ 2 \sigma_t^{j,2}(Q^{1,j}d B_t^{2,1} + Q^{2,j} d B_t^{2,2}),
\end{equation}
from which:
\begin{equation}
d w^j_t:= \frac{ d \Sigma_t^{j,j} -F^j_t dt}{ 2 \sqrt { \Sigma_t^{j,j}((Q^{1,j})^2+ (Q^{2,j})^2) }} = \frac{ \sigma_t^{1,j} (Q^{1,j}d B_t^{1,1} +Q^{2,j}d B_t^{1,2})+  \sigma_t^{j,2} (Q^{1,j}d B_t^{2,1} +Q^{2,j}d B_t^{2,2})  }{ \sqrt {\Sigma_t^{j,j}((Q^{1,j})^2+ (Q^{2,j})^2) }}.
\end{equation}
By taking the quadratic variation of the right-hand side we see that $w^j_t$ are two Brownian motions such that $d\Sigma_t^{j,j}= F^j_t dt+ 2 \sqrt {\Sigma_t^{j,j}((Q^{1,j})^2+ (Q^{2,j})^2) } dw^j_t$; equations (\ref{activitycorr}) and (\ref{wishcorr}) then follow from a direct computation.

Since $X_{U_t^d}$ is orthogonal to every entry of the matrix Brownian motion $B_t$, the change in the dynamics of  $\Sigma_t$ under $\mathbb Q(z,w)$ is only due to the correlation between $X^c_t$ and $B_t$. Hence, for $(z,w) \in \Theta$, the Radon-Nikodym derivative $M_t$ to be considered in (\ref{measurechange}) reduces to
 
\begin{equation}
M_t=\mathcal E \left(iz  \int_0^t \sqrt{ \Sigma_s^{1,1} } d X^c_s \right).
 \end{equation}
  Furthermore, for $j=1,2$ we have: 
 \begin{align} d \left \langle \int_0^{\cdot} \sqrt{ \Sigma_s^{1,1}} d X_s^c,  B^{j,1} \right \rangle_t & = \rho  \sigma^{1,j}_t dt \\d  \left \langle \int_0^{\cdot} \sqrt{ \Sigma_s^{1,1}}d X_s^c,  B^{j,2} \right \rangle_t & = 0
\end{align}
so that application of Girsanov's theorem tells us that 
\begin{equation} d \tilde  B_t = d B_t - iz \rho \left( \begin{matrix}  \displaystyle{    \sigma^{1,1}_t dt }& 0  \\ \displaystyle{   \sigma^{1,2}_t dt } & 0  \end{matrix}   \right)
\end{equation}
is a $\mathbb Q(z,w)$-matrix Brownian motion. Solving the above for $B_t$ and substituting in (\ref{Wishart}) yields (\ref{Wishartz}). Equation (\ref{WishartPhi}) then follows from (\ref{generallaplace}).

Finally, we give the formula for $\mathcal L_{T_t, U_t}(\cdot)$. For $\tau >0$ and $n>1$ consider the transform:  
\begin{equation}\label{integrwishtransf}
 \phi_\Sigma(z)=\e \left[\exp\left(-\int_0^\tau \sum_{j=1}^n  z_j \Sigma_s^{j,j}ds \right) \right]
\end{equation}
 for every vector of complex numbers $z=(z_1, \ldots, z_n)$ such that the above expectation is  finite. The function $\phi_\Sigma (z)$ is exponentially-affine of the form 
\begin{equation}\label{expaffinewish}
 \phi_\Sigma(z)=\exp(-a(\tau)-Tr(A(\tau) \Sigma_0)),
\end{equation} 
 since it is a particular case of the transforms studied in e.g. Grasselli and Tebaldi (2007), 
  and Gouri\'eroux  (2003).  The ODEs for $A(\tau), a(\tau)$ are given by:
\begin{equation}
\begin{array}{rl}\label{wishode1} \displaystyle{A(\tau)'= A(\tau)M+M^T A(\tau) - 2 A(\tau) Q^T Q A(\tau) +D, } &  A(0)=0 \end{array} 
\end{equation}
\begin{equation}\label{wishode2} 
\begin{array}{rl}\displaystyle{a(\tau)'= Tr(c Q^T Q A(\tau) ),}  &  a(0)=0. \end{array}
\end{equation}
Here $D$ is the diagonal matrix having the values $z_1, \ldots , z_n$ on the diagonal. The solution of (\ref{wishode1})-(\ref{wishode2})  is obtainable through a linearization procedure that entails  doubling the dimension of the problem, which yields:
\begin{equation}\label{integrwishtransf1} A(\tau)=(A^{2,2}(\tau))^{-1}A^{2,1}(\tau)
\end{equation}
\begin{equation}\label{integrwishtransf2} a(\tau)= \frac{c}{2} Tr(\log\left(A^{2,2}(\tau)\right) +M^T \tau)
\end{equation}
\begin{equation}\label{integrwishtransf3}\left(\begin{matrix} A^{1,1}(\tau) & A^{1,2}(\tau) \\ A^{2,1}(\tau) & A^{2,2}(\tau)  \ 
\end{matrix}\right) = \exp \left( \tau  \left( \begin{matrix} M  & 2 Q^T Q \\ D & -M^T  \end{matrix}\right) \right)
\end{equation}
(see for example  Gouri\'eroux  2003, proposition 7, or Grasselli and Tebaldi 2007, section 3.4.2). The formula for $\mathcal L_{\Delta T, \Delta U }$  follows from (\ref{integrwishtransf1})- (\ref{integrwishtransf3}) when we choose $n=2$,  $(z_1, z_2)=(z,w)$ in (\ref{integrwishtransf}), and set $\tau=t-t_0$, $ \Sigma_0=\Sigma_{t_0}$ in (\ref{expaffinewish}).  \qed
\end{proof}

\bibliographystyle{natbib}

\section*{Tables and figures}

\begin{table}

\caption{ \small{{Parameters from the S\&P estimations of Fang (2000), section 4. }}}
\vspace{10pt}
\begin{tabular}{ llllll }
\hline\noalign{\smallskip}
    Parameters  & Black-Scholes & Heston  & Merton & Bates  & Fang   \\ 
  \hline

 $\sigma_{t_0}$     & 0.14    & 0.15   & 0.12  & 0.15  & 0.14   \\  
\hline 
 $\alpha$           &  & 4.57   &       &  8.93  & 6.5 \\
\hline
$\theta$            &  & 0.0306 &       &  0.0167 &  0.0104 \\
\hline
 $\eta$              &  &  0.48  &       & 0.22   & 0.2 \\
\hline
 $\rho$              &  & -0.82  &       & -0.58 & -0.48  \\
\hline
  $\lambda_0$     &      &        & 1.42  & 0.39  & 0.41       \\   
\hline
 $\delta$           &  &        &   0.0894 &  0.1049 & 0.2168   \\
\hline
$\kappa$             &     &        &  -0.075 &  -0.11 & -0.21  \\
\hline
$\alpha_{\lambda}$  &    &        &      &      & 5.06  \\ 
\hline
 $\theta_{\lambda}$  &   &        &      &      &  0.13  \\ 
\hline
 $\eta_{\lambda}$     &  &        &      &      & 1.069  \\
\hline\noalign{\smallskip}
\end{tabular}

\end{table}

\begin{table}

\caption{ \small{{ $S_{t_0}=100$, $K=H=80$, $Q=0.05$, $t_0=0$, $t=1$, $r=0.06$, $\target=0.1$, $TV_{t_0}=0$.}}}
\begin{tabular}{  llllllll }
\hline\noalign{\smallskip}
  \multicolumn{1}{l}{Model} & \multicolumn{2}{l} {Vanilla Call}  &   \multicolumn{2}{l}{Volatility Call}  & \multicolumn{2}{l}{TVO Call}   \\  
\cline{2-7}

  \multicolumn{1}{l}{} & \multicolumn{1}{l} {AV}  &   \multicolumn{1}{l} {MC}  & \multicolumn{1}{l}{AV} &   \multicolumn{1}{l} {MC} &  \multicolumn{1}{l}{AV} &  \multicolumn{1}{l} {MC}     \\
\hline

B-S                &     24.7627  &   24.7775(0.05\%)    &   0.0847 &   0.0848(0.12\%) &   17.5441  &   17.6982(0.87\%)  \\
\hline
Heston  &    25.3893 &  25.3710(0.07\%) & 0.1088 &  0.1084(0.37\%) & 17.2248 &  17.6044(2.16\%)  \\
\hline
Merton                  &   25.3243 & 25.2290(0.38\%)  &     0.1192 &  0.1194(0.17\%)  &   17.7529 & 17.7922(0.22\%)    \\  
\hline
Bates                  &   25.1166   &   25.0889(0.11\%)&    0.1002 &  0.1005(0.30\%) &  18.5980  &   18.7480(0.80\%) \\
\hline
Fang                &     25.5686  &    25.6508(0.32\%)    &   0.0907 &   0.0892(1.68\%) &   24.0494  &   24.0764(0.11\%)  \\
\hline\noalign{\smallskip}
\end{tabular}

\vspace{-.1cm}

\end{table}

\begin{table}

\caption{ \small{{ $S_{t_0}=100$, $K=H=120$, $Q=0.1$, $t_0=0.5$, $t=4$, $r=0.039$, $\target=0.1$, $TV_{t_0}=0.018$.}}}
\begin{tabular}{  llllllll }
\hline\noalign{\smallskip}
  \multicolumn{1}{l}{Model} & \multicolumn{2}{l} {Vanilla Call}  &   \multicolumn{2}{l}{Volatility Call}  & \multicolumn{2}{l}{TVO Call}   \\  
\cline{2-7}

  \multicolumn{1}{l}{} & \multicolumn{1}{l} {AV}  &   \multicolumn{1}{l} {MC}  & \multicolumn{1}{l}{AV} &   \multicolumn{1}{l} {MC} &  \multicolumn{1}{l}{AV} &  \multicolumn{1}{l} {MC}     \\
\hline

B-S                &  8.4801  &  8.4784(0.02\%)    &   0.1672 &  0.1695(1.36\%) &   5.7622 &   5.6957(1.17\%)  \\
\hline
Heston  & 10.3063 & 10.3023(0.04\%) & 0.2167 &  0.2172(0.23\%) &  6.3815&  6.7080(4.87\%) \\
\hline
Merton                 &   11.5845    &    11.5713(0.11\%)  &   0.2357 & 0.2356(0.04\%) &   7.4564 &  7.4239(0.44\%)
 \\
\hline
Bates                  &      9.8607  &   9.8371(0.24\%) &  0.2002 & 0.2001(0.05\%)  &  6.8180 & 6.9085(1.31\%)    \\
\hline
Fang                &    8.8630  &  8.8737(0.12\%)  &  0.1827 & 0.1828(0.05\%)    &     7.4173 &  7.5046(1.16\%) \\
\hline\noalign{\smallskip}
\end{tabular}

\vspace{-.1cm}
\end{table}

\begin{table}

\caption{ \small{{ $S_{t_0}=100$, $K=H=100$, $Q=0.25$, $t_0=1.25$, $t=1.5$, $r=0.072$, $\target=0.1$, $TV_{t_0}=0.23$.}}}
\begin{tabular}{  llllllll }
\hline\noalign{\smallskip}
  \multicolumn{1}{l}{Model} & \multicolumn{2}{l} {Vanilla Call}  &   \multicolumn{2}{l}{Volatility Call}  & \multicolumn{2}{l}{TVO Call}   \\  
\cline{2-7}

  \multicolumn{1}{l}{} & \multicolumn{1}{l} {AV}  &   \multicolumn{1}{l} {MC}  & \multicolumn{1}{l}{AV} &   \multicolumn{1}{l} {MC} &  \multicolumn{1}{l}{AV} &  \multicolumn{1}{l} {MC}     \\
\hline

B-S               &     3.7627  &   3.7346(1.02\%)    &   0.2300 &   0.2305(0.22\%) & 0.9771&   0.9437(3.54\%)  \\
\hline
Heston  &   4.1390 &   4.1304(0.21\%) & 0.2318  &  0.2320(0.09\%) & 1.0480 & 1.0451(0.28\%)
  \\
\hline
Merton                  &    4.4169   &    4.4435(0.60\%) &   0.2348 &     0.2343(0.21\%) &  1.1254 &  1.1235(0.17\%)  \\
\hline
Bates                  &   4.1842    &   4.1687(0.37\%)  & 0.2327 & 0.2328(0.04\%)   &  1.0593 &  1.0544(0.46\%)
   \\
\hline
Fang                &  4.3219  & 4.3420(0.46\%) &  0.2362  & 0.2362(0.00\%)   &  1.0919 & 1.0987(0.62\%)  \\
\hline\noalign{\smallskip}
\end{tabular}
\vspace{-.1cm}

\end{table}

\begin{table}

\caption{ \small{{ $S_{t_0}=100$, $K=H=60$, $Q=0.2$, $t_0=3$, $t=5$, $r=0.0225$, $\target=0.1$, $TV_{t_0}=0.19$.}}}
\begin{tabular}{  llllllll }
\hline\noalign{\smallskip}
  \multicolumn{1}{l}{Model} & \multicolumn{2}{l} {Vanilla Call}  &   \multicolumn{2}{l}{Volatility Call}  & \multicolumn{2}{l}{TVO Call}   \\  
\cline{2-7}

  \multicolumn{1}{l}{} & \multicolumn{1}{l} {AV}  &   \multicolumn{1}{l} {MC}  & \multicolumn{1}{l}{AV} &   \multicolumn{1}{l} {MC} &  \multicolumn{1}{l}{AV} &  \multicolumn{1}{l} {MC}     \\
\hline
B-S &     42.6506 &    42.6452(0.01\%)    &   0.2670&   0.2665(0.19\%) &  19.7252 &   19.9181(0.96\%)  \\
\hline
Heston  &  42.9595  & 43.0010(0.10\%)  &   0.2859 &  0.2858(0.03\%) &    19.8454 &  19.6512(0.99\%) \\
\hline
Merton            &  42.8984    & 42.8580(0.09\%)   &  0.2955 &   0.2954(0.03\%)&  19.4192 &   19.3975(0.11\%) \\
\hline
Bates                  &   42.7768   & 42.7928(0.04\%) & 0.2804 & 0.2802(0.07\%)  &  19.8042  & 19.8318(0.14\%)    \\
\hline
Fang                & 43.0039    &   43.0252(0.05\%)  &  0.2793 &  0.2791(0.07\%)  &   20.5992 & 20.5998(0.01\%)     \\
\hline\noalign{\smallskip}

\end{tabular}

\vspace{-.1cm}

\end{table}

\begin{table}

\caption{ \small{{ $S_{t_0}=100$, $K=H=130$, $Q=0.015$, $t_0=1$, $t=2.5$, $r=0.087$, $\target=0.1$, $TV_{t_0}=0.009$.}}}
\begin{tabular}{  llllllll }
\hline\noalign{\smallskip}
  \multicolumn{1}{l}{Model} & \multicolumn{2}{l} {Vanilla Call}  &   \multicolumn{2}{l}{Volatility Call}  & \multicolumn{2}{l}{TVO Call}   \\  
\cline{2-7}

  \multicolumn{1}{l}{} & \multicolumn{1}{l} {AV}  &   \multicolumn{1}{l} {MC}  & \multicolumn{1}{l}{AV} &   \multicolumn{1}{l} {MC} &  \multicolumn{1}{l}{AV} &  \multicolumn{1}{l} {MC}     \\
\hline

B-S   &     2.3393 &    2.3080(1.36\%)    &   0.1590&   0.1588(0.13\%) &   1.9535 &   1.8622(4.90\%)  \\
\hline
Heston    & 2.5098 &  2.5071(0.11\%)    & 0.1852  & 0.1862(0.54\%)  & 2.2190 &  2.1317(4.10\%) \\
\hline
Merton          &  3.7078      &    3.6843(0.64\%) &  0.1983   & 0.1981(0.10\%)  &   3.0330 & 3.0165(0.55\%)  \\
\hline
Bates                  &    2.7416   & 2.7380(0.13\%)   &   0.1767  &   0.1769(0.11\%) &   2.3727 &  2.3798(0.30\%)   \\
\hline
Fang               &   1.9814  &  1.9410(2.08\%)  &  0.1664 &  0.1668(0.24\%)  &   1.9453 &    1.9167(1.49\%)     \\
\hline\noalign{\smallskip}

\end{tabular}

\vspace{-.1cm}

\end{table}

\begin{figure}
\centering
\includegraphics[scale=0.59]{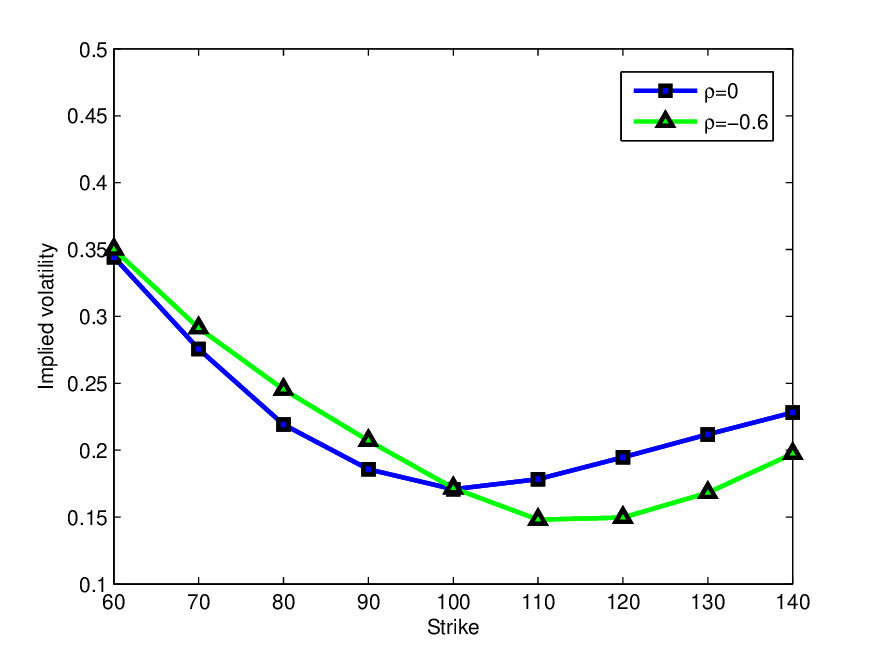}
\caption{\small{Comparison of the 3-month volatility skew in the Bates model for two different values of $\rho$. There is only a small difference in the skew of the two curves.}}
\label{fig1}
\end{figure}

\begin{figure}
\centering
\includegraphics[scale=0.59]{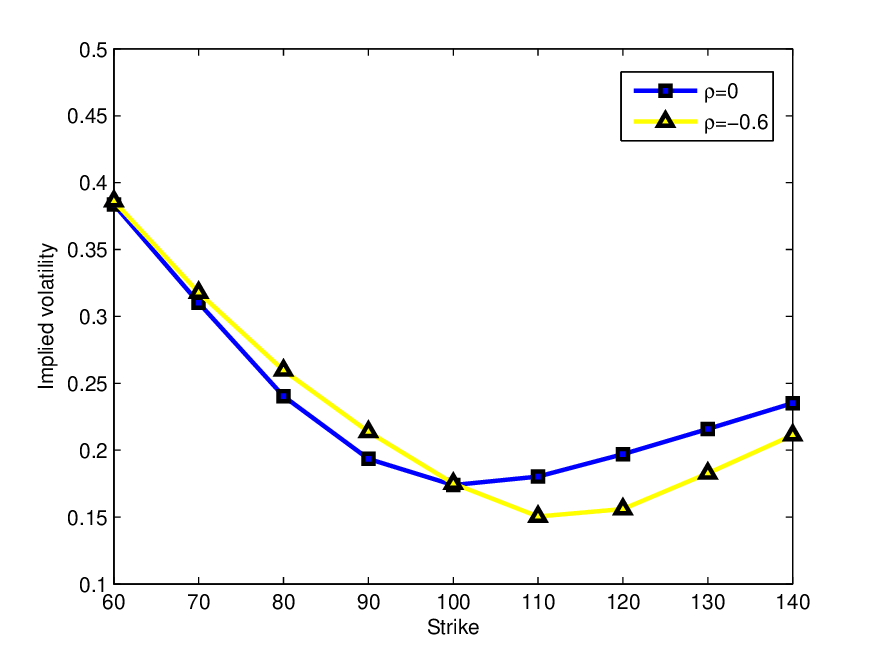}
\caption{\small{Comparison of the 3-month volatility skew in the Fang model for two different values of $\rho$.  The situation is very similar to that of figure \ref{fig1}.}}
\label{fig2}
\end{figure}

\begin{figure}
\centering
\includegraphics[scale=0.59]{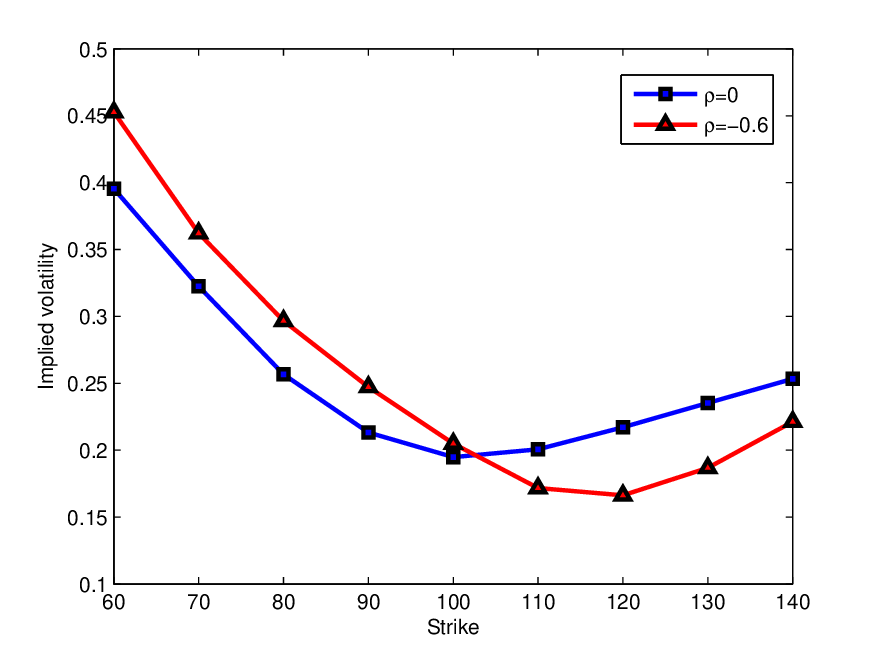}
\caption{\small{Comparison of the 3-month volatility skew in the CSVJA model for two different values of $\rho$. The negative skew increases much more than in the other two models. By equation (\ref{wishcorr}), a value  $R_1=-0.7$ must be used in (\ref{matrices}) in order to obtain $\rho=-0.6$.}}
\label{fig3}
\end{figure}

\end{document}